\documentclass[11pt]{article}
\usepackage[utf8]{inputenc}

\usepackage[letterpaper,top=2.54cm,bottom=2.54cm,left=2.54cm,right=2.54cm,marginparwidth=1.75cm]{geometry}

\usepackage[parfill]{parskip}
\usepackage{natbib}
\usepackage{xr-hyper,refcount}
\usepackage{graphicx}
\usepackage[utf8]{inputenc} %
\usepackage[colorlinks,citecolor=blue,urlcolor=blue,linkcolor=blue,bookmarks=false]{hyperref}       %
\usepackage{url}            %
\usepackage{booktabs}       %
\usepackage{amsfonts}       %
\usepackage{nicefrac}       %
\usepackage{bbold}          %
\usepackage{microtype}      %
\usepackage[parfill]{parskip}
\usepackage{algorithm,algorithmic}
\usepackage{amsmath,amsthm,amssymb,bbm}
\usepackage{mathtools}
\usepackage{cases}
\usepackage{comment}
\usepackage{subcaption}
\usepackage{color}
\usepackage{appendix}
\usepackage{xspace}
\usepackage{enumitem}
\usepackage[english]{babel}
\usepackage{authblk}   %
\usepackage{natbib}
\usepackage{graphicx}
\usepackage{amssymb}

\usepackage{enumitem}
\usepackage{booktabs}

\usepackage{leqisymbols}
\usepackage{bm}
\usepackage{xcolor}
\usepackage{setspace}

\title{Median Optimal Treatment Regimes}
\author[1]{Liu Leqi}
\author[2]{Edward H. Kennedy}
\affil[1]{Machine Learning Department}
\affil[2]{Department of Statistics \& Data Science}
\affil[ ]{Carnegie Mellon University}
\affil[ ]{\texttt{\href{mailto:leqil@cs.cmu.com}{\textcolor{black}{leqil@cs.cmu.edu}}, \href{mailto:edward@stat.cmu.com}{\textcolor{black}{edward@stat.cmu.edu}} }}
\date{}

\begin{document}
\maketitle

\begin{abstract}
Optimal treatment regimes are personalized policies 
for making a treatment decision 
based on subject characteristics, 
with the policy chosen to maximize some value. 
It is common to aim to 
maximize the mean outcome in the population, 
via a regime assigning treatment only to 
those whose mean outcome 
is higher under treatment versus control. 
However, the mean can be an unstable measure of centrality, 
resulting in imprecise statistical procedures, 
as well as unrobust decisions that 
can be overly influenced by a small fraction of subjects. 
In this work, we propose a new median optimal treatment regime that instead treats individuals whose \emph{conditional median} is higher under treatment.
This ensures that optimal decisions 
for individuals from the same group are not overly influenced 
either by (i) a small fraction of the group (unlike the mean criterion), 
or (ii) unrelated subjects 
from different groups (unlike marginal median/quantile criteria). 
We introduce a new measure of value, the Average Conditional Median Effect (ACME), which summarizes across-group median treatment outcomes of a policy, and 
which the median optimal treatment regime maximizes. 
After developing key motivating examples that distinguish 
median optimal treatment regimes from mean and marginal median optimal treatment regimes, 
we give a nonparametric efficiency bound 
for estimating the ACME of a policy, 
and propose a new doubly robust-style estimator 
that achieves the efficiency bound under weak conditions. 
\edit{To construct the median optimal treatment regime, we introduce a new doubly robust-style estimator for the conditional median treatment effect.} 
Finite-sample properties %
are explored via numerical simulations
and the proposed algorithm is illustrated 
using data from a randomized clinical trial in patients with HIV.

\end{abstract}

\section{Introduction}
\label{sec:introduction}
Optimal treatment regime methods 
aim to learn policies 
that map subject covariates to a decision, 
typically with the goal of maximizing a population criterion. 
The task of assigning treatment decisions
based on individual characteristics  
is common and crucially important in many disciplines, with
applications ranging from personalized medicine to loan approval  
to educational program assignment. By now there is a large literature on theory, methods, and applications of optimal treatment regimes \citep{murphy2003optimal, robins2004optimal, laber2014dynamic, chakraborty2013statistical, schulte2014q, van2014targeted, kosorok2019precision}; we mostly refer to this and related work for more general background and details. 

As more individualized policies are deployed in practice, 
it is of particular importance to examine the target value %
for which the treatment regimes are optimized. 
Traditionally, the mean outcome in the population has been used as the criterion. An important reference in this stream for our purposes is~\citet{van2014targeted}, which studies doubly robust estimation of the mean outcome under both a known policy and the unknown optimal policy, showing how nonparametric efficiency bounds can be achieved in both cases. This is part of a larger literature on semiparametric efficiency theory and doubly robust estimation (also known as one-step or debiased or double machine learning estimation)~\citep{newey1990semiparametric, bickel1993efficient, van2002semiparametric, van2003unified, tsiatis2007semiparametric, kennedy2016semiparametric, chernozhukov2018double}.

However, compared to the median,
means can be a poor and/or sensitive measure of  the centrality of 
a distribution, for example in the presence of skew or contamination \citep{huber1967behavior, casella2002statistical, huber2004robust}.
As a result, a recent literature has evolved studying alternative target values for learning policies. %
\citet{linn2017interactive,wang2018quantile, luedtke2020performance} propose 
the marginal quantile to be the criterion in settings where the outcome distribution is skewed or the tail of the outcome distribution is of interest.  
For simplicity we refer to these approaches as marginal median-based, since the issues we detail with marginal median objectives are the same as those for generic quantiles. 
In presenting the marginal quantile optimal treatment regime, \citet{linn2017interactive} also considers 
a marginal cumulative distribution function related target value 
where the goal 
is to find a policy that maximizes the probability
of the outcome being above a given threshold.
\citet{qi2019estimating} uses conditional value-at-risk~\citep{rockafellar2002conditional} on the outcome distribution as the criterion to ensure the policy to be risk-averse. 
Unfortunately, in addition to not having a closed-form optimal policy,  the 
 marginal median approaches %
yield optimal treatment decisions for subjects with covariates $X=x$ that depend on outcomes of \emph{other subjects} with different covariates $X \neq x$. We view this as allowing a form of unfairness, which will be discussed in  more detail shortly.

Motivated by these concerns, 
in this paper we 
propose a treatment regime which assigns  treatment to those subjects who have a higher \emph{conditional median} outcome under treatment versus control, and which is optimal with respect to a new objective we call the Average Conditional Median Effect (ACME).
Crucially, the optimal policy for the ACME has two key properties, which in general do not both hold for either mean or marginal median optimal treatment regimes:
\begin{itemize}
    \item ``Within-group robustness'': Within a group, i.e., for subjects with the same observed covariates, the treatment decision is based on the median of the conditional outcome distribution, and thus the policy is robust to outliers (e.g., when a small fraction of the group has extreme outcomes), unlike the mean optimal treatment regime. 
    \item ``Across-group fairness'': Across groups, i.e., for subjects with different observed covariates, the treatment decision for a given group cannot depend on outcomes of a different group, unlike the marginal median optimal treatment regime.
\end{itemize}

\begin{remark}
Here we use the term ``fair'' rather loosely; in contrast
there has been lots of recent work defining fairness more formally 
\citep{dwork2012fairness, hardt2016equality, chouldechova2018frontiers,mehrabi2019survey}. 
We also acknowledge that other definitions of fairness could be plausible in our setup, 
and really only use the above for a convenient shorthand. 
We describe in much more detail why we refer to the second property as 
fairness-related in an example in Section~\ref{sec:motivating-example}.
\end{remark}

Our main contributions are as follows. 
We reflect on existing standard target values %
and their optimal treatment regimes 
and 
propose the median optimal treatment regime\footnote{We use \emph{marginal} median optimal treatment regimes 
to refer to policies that maximize marginal median values.},  
that assigns treatment to an individual based on their 
\emph{conditional} median treatment effect (Section~\ref{sec:acme}). 
A new measure of policy value, namely the average conditional median effect, is defined. 
The proposed regime promotes within-group robustness and across-group fairness, 
in comparison to the mean optimal 
and marginal median optimal policies 
(Section~\ref{sec:policy_comparision}). 
In Section \ref{sec:efficiency}, we establish the local asymptotic minimax bound for estimating the ACME, 
and construct a doubly robust-style estimator along with a simple algorithm 
that achieves the bound under mild conditions (Section~\ref{sec:estimation}). 
\edit{
To learn the median optimal treatment regime,
we propose a new doubly robust-style estimator for the Conditional Median Treatment Effect (CMTE) in Section~\ref{sec:learning-gamma}. 
}
Finally,
we use numerical simulations to show finite-sample properties of the estimator 
and illustrate the algorithm using a dataset from a randomized clinical trial on HIV patients (Section~\ref{sec:experiment}).

\section{Preliminaries}
\label{sec:preliminaries}
We are given independent and identically distributed (iid) samples $Z_i = (X_i, A_i, Y_i)$ drawn from a distribution $\mathbb{P}$ 
where $X_i \in \mathcal{X} \subseteq \mathbb{R}^d$ are covariates, 
$A_i \in \{0, 1\}$ is a binary treatment assignment 
and $Y_i \in \mathcal{Y}$ is a continuous real-valued outcome of interest. 
We use $Y^a$ to denote the potential outcome under treatment $A=a$~\citep{rubin1974estimating}. Throughout we let $m\left(V \mid w\right) = \inf\{v \in \mathbb{R}: \mathbb{P}(V \leq v \mid W=w) \geq 1/2\}$ 
denote the median of a generic random variable $V$ given $W=w$. To simplify the presentation, we introduce the following notation for components of the distribution $\mathbb{P}$:
\begin{align*}
\pi_a(x) &= \mathbb{P}(A=a \mid X=x), \\
F_a(y \mid x) &=\mathbb{P}(Y \leq y \mid X=x, A=a), \\
f_a(y \mid x) &= \frac{d}{d y} F_a( y \mid x), \\
m_a(x) &= m\left(Y \mid X=x, A=a\right), \\
F_{a,m}(x) &= F_a(m_a(x) \mid X=x), \\
f_{a,m}(x) &= f_a(m_a(x) \mid X=x), \\
\sigma_a(x) &= \sqrt{\text{Var}(Y \mid X=x, A=a)}.
\end{align*}
We note that $\pi_1(x)$ is called the propensity score, i.e., chance of receiving treatment given covariates,
and $\pi_0(x) = 1 - \pi_1(x)$. 
We assume $F_a(y\mid x)$ to be absolutely continuous 
and hence $f_a(y\mid x)$ exists. 
Throughout the paper, we refer to 
$F_a(y \mid x)$ as the conditional cumulative distribution function, 
$f_a(y \mid x)$ as the conditional density, 
$m_a(x)$ as the conditional median,
$f_{a,m}(x)$ as the conditional density at the conditional median, 
and $\sigma_a(x)$ as the conditional standard deviation. 
We assume the following standard conditions for identifying causal effects: 
\allowdisplaybreaks
\begin{assumption}
The following causal assumptions hold  %
\begin{align*}
    (\text{Consistency})\qquad &Y=A Y^1 + (1-A ) Y^0, %
    \\
    (\text{Positivity})\qquad &\mathbb{P}\left(\epsilon \leq \pi_1(X) \leq 1 - \epsilon \right) = 1 
    \text{ for some } \epsilon > 0,
    \\
    (\text{Exchangeability})\qquad &A \indep Y^a \mid X \text{ for }
    a = 0,1. %
\end{align*}
\label{assump:causal}
\end{assumption}

\paragraph{Additional Notation}
We use $\mathbb{P} \widehat{f} := \int \widehat{f}(z) d \mathbb{P}(z)$ 
to denote the expectation conditioned upon the randomness of $\widehat{f}$  
and $\mathbb{P}_n f := \mathbb{P}_n \{f(Z)\} = \frac{1}{n} \sum_{i=1}^n f(Z_i)$. 
We denote the Euclidean norm for a real-valued vector $x \in \mathbb{R}^d$ to be $\|x\|_2$ 
and the $L_2(\mathbb{P})$ norm of a function $f$ to be $\|f\| := \left(\int f(z)^2 d \mathbb{P}(z) \right)^{1/2}$. 
Finally, $a \lesssim b$ is equivalent to $a \leq C b$ for some positive constant $C$.

\section{Proposed Target Policy \& Value} 
\label{sec:acme}

In this section, we first review common target values
and their optimal policies 
in the current literature (Section~\ref{sec:standard_policy}). 
In Section~\ref{sec:median_otr},
we introduce a new policy and a new measure of a policy's value, 
based on conditional medians. 
We show that under Assumption~\ref{assump:causal} 
these policies and values are identified. 

\subsection{Standard Policies \& Values}
\label{sec:standard_policy}

Let $d: \mathcal{X} \to \{0,1\}$ denote a deterministic policy
that assigns a treatment based on input covariates 
and $\mathcal{D}$ be the set of all such measurable functions.
Traditionally, most optimal treatment regime research~\citep{murphy2003optimal,zhang2012robust} has 
focused on finding the policy that maximizes the mean outcome 
$$
\mathbb{E}[Y^d] := \mathbb{E}[d(X)Y^1 + \left(1-d(X)\right) Y^0],$$ 
which under Assumption~\ref{assump:causal} is equivalent to 
$
    \mathbb{E}[Y^d] = \mathbb{E}_X\left[\mu_d(X)\right], 
$
where $\mu_d(X) := d(X) \mu_1(X) + (1-d(X)) \mu_0(X)$ 
and $\mu_a(X) := \mathbb{E}[Y\mid X,A=a]$. 
When the goal is to maximize the overall population mean $\mathbb{E}[Y^d]$,  %
it is well-known that the optimal policy is given by%
$$
     d_\text{MEAN}^*(X) %
        := \mathbb{1}\{\mu_1(X) > \mu_0(X)\}.
$$
This policy simply treats only those subjects whose conditional mean treatment effect is positive. 

\begin{remark} \label{rem:condval}
Often, the mean optimal policy $d_\text{MEAN}^*$ is motivated by the fact that it maximizes $\mathbb{E}[Y^d]$; however, we view its primary motivation as coming from the fact that it maximizes the \emph{conditional} value $\mathbb{E}[Y^d\mid X=x]$ for all $x \in \mathcal{X}$, i.e., it is mean optimal for subjects of any type. From this perspective, the \emph{marginal} value $\mathbb{E}[Y^d]$ is just a one-number summary of the overall performance of a policy, %
across a heterogeneous population of subjects with different covariates. As we discuss further in subsequent sections, we view the maximization of the conditional value as more fundamental and fair in the across-group sense introduced in Section~\ref{sec:introduction} and detailed further in Section~\ref{sec:motivating-example}.  
\end{remark}

While the mean is a valuable measure of centrality in many cases, 
it is sensitive to outliers and so can be %
an inappropriate target value 
when some have extreme responses to treatment. 
This has led to recent work maximizing the marginal median 
$$ 
m(Y^d) := m\left( d(X)Y^1 + \left(1-d(X)\right) Y^0\right).
$$ 
In the general case, 
marginal quantile-based values are proposed where the goal is to 
find the optimal policy with respect to 
a quantile 
of the outcome~\citep{linn2017interactive, wang2018quantile, kallus2019localized}. 
Under Assumption~\ref{assump:causal}, we obtain that 
\begin{align*}
        m(Y^d) &= 
        \inf\left\{ m: \mathbb{P}(Y^d \leq m) \geq 1/2 \right\}\\
        &= \inf\left\{ m: \int \mathbb{P}(Y^d \leq m \mid  X=x) d\mathbb{P}(x) \geq 1/2 \right\}\\
        &=\inf\left\{ m: \int \mathbb{P}(Y \leq m \mid A=d(x), X=x) d\mathbb{P}(x) \geq 1/2 \right\}.
\end{align*}
From here on we use $m(Y^d)$ to refer to the identified quantity above, with the understanding that it corresponds the counterfactual marginal median only under Assumption \ref{assump:causal}.

Unfortunately, as illustrated in Proposition~\ref{prop:GoM}, %
the marginal median objective  
allows specific subjects' optimal treatment assignments to be determined by \emph{different} subjects, with different covariate values. 
We see this as a violation of across-group fairness. For example, if in some population one group suddenly started responding more to treatment, then under the marginal median objective, the optimal treatment for other groups could change -- even if their response to treatment did not.

Unlike the mean optimal policy $d_{\text{MEAN}}^*(x)$, 
in general, the optimal policy under the marginal median objective (which we denote as $d^*_\text{MME}(x)$) cannot be viewed as maximizing a parameter defined 
upon the conditional outcome distributions $[Y\mid X=x, A=a]$ for $a \in \{0,1\}$, and  
does not have a closed form that depends only on the conditional outcome distributions. %

\subsection{Median Optimal Treatment Regimes}
\label{sec:median_otr}

As mentioned in Remark~\ref{rem:condval}, we believe conditional values %
are generally more fundamental and useful than marginal values. Therefore 
here we propose a policy $d^*_\text{ACME}$ that maximizes the conditional median $m(Y^d\mid X=x)$ for all $x \in \mathcal{X}$. 
This implies that 
\begin{align}
    d_{\text{ACME}}^*(X):= \mathbb{1}\{m_1(X) > m_0(X)\}.
\label{eq:d_acme}
\end{align}
Using the conditional median instead of the conditional mean to measure centrality of the conditional outcome distributions, 
compared to the mean optimal policy $d^*_\text{MEAN}$, $d^*_\text{ACME}$ is more \emph{robust} to outliers. 
Unlike the marginal median optimal policy $d^*_\text{MME}$, for all $x \in \mathcal{X}$, our proposed policy $d^*_\text{ACME}(x)$  is more \emph{individualized}, in the sense that optimal treatment assignments for subjects with $X=x$ do not depend on outcomes of different subjects with $X \neq x$. 
We summarize the overall performance of our proposed conditional median optimal policy
(also denoted as median optimal treatment regime and median optimal policy)
$d^*_\text{ACME}$ 
across groups of subjects with different covariates using 
the average conditional median effect, which we introduce below. 
\begin{definition} Given a policy $d \in \mathcal{D}$, we define 
the average conditional median effect (ACME) as
\begin{align}
    \mathbb{E}_X[m(Y^d \mid  X)] := \mathbb{E}_X\left[m\left(d(X) Y^1 + \left(1 - d(X)\right) Y^0 \mid  X\right)\right]. 
\label{eq:ACME}
\end{align}
\end{definition}

\begin{remark}
Traditionally a treatment ``effect'' is a \emph{contrast} between (distributions of) potential outcomes; for simplicity we call the ACME parameter an effect even though it involves counterfactuals under only one treatment policy, and so is not a contrast. 
\end{remark}

\paragraph{Identification}
Under Assumption~\ref{assump:causal}, we have that 
\begin{align*}
    \mathbb{E}_X[m(Y^{d} \mid X)]
    &= \mathbb{E}_X[
    d(X)m(Y^1\mid A=1, X) 
    + (1-d(X))m(Y^0\mid A=0, X)]\\
    &= \mathbb{E}_X[
    d(X)m(Y\mid A=1, X) 
    + (1-d(X))m(Y\mid A=0, X)]\\
    &= \mathbb{E}_X\left[m_d(X)\right],
\end{align*}
where $m_d(X) := d(X)m_1(X) + \left(1 - d(X) \right) m_0(X)$. %
For the rest of the paper, without further specification, 
we assume Assumption~\ref{assump:causal} to hold. 

As illustrated at the beginning of Section~\ref{sec:median_otr}, 
the ACME is used to summarize the performance of the conditional median optimal policy $d_{\text{ACME}}^*$.
We show here that, just as  $d_{\text{MEAN}}^*$ optimizes $\mathbb{E}[\mathbb{E}[Y^d\mid X]]$, $d_{\text{ACME}}^*$ optimizes the ACME $\mathbb{E}[m(Y^d\mid X)]$, i.e., 
for all $d \in \mathcal{D}$, 
\begin{align*} 
    \mathbb{E}_X[m(Y^{d} \mid X)] 
    &=  \mathbb{E}_X[d(X) m_1(X) + (1-d(X)) m_0(X)]\\
    &= \mathbb{E}_X[m_0(X) + \left(m_1(X) - m_0(X) \right)d(X)]\\
    &\leq \mathbb{E}_X[m_0(X) + \left(m_1(X) - m_0(X) \right)d_\text{ACME}^*(X)].
\end{align*}

\begin{remark}
We note that in the case where the polices can only be measurable functions of a subset of covariates $V \subseteq X$, 
the optimal treatment regime for ACME
is not as straightforward as the one for the mean outcome, due to the nonlinearity of the median. 
Developing the $V$-specific optimal treatment regime for ACME is an avenue for future work. 
\end{remark}

\section{Policy Comparisons}
\label{sec:policy_comparision}

In this section, we compare the three policies 
$d^*_\text{MEAN}$, $d^*_\text{MME}$ and $d^*_\text{ACME}$. 
In Section~\ref{sec:d_ate_d_acme}, 
we give a simple condition under which the 
mean and median optimal treatment regimes are equivalent, and point out that the mean optimal policy is not necessarily robust within groups (defined in Section~\ref{sec:introduction}).
In Section~\ref{sec:d_acme_d_mte}, 
we illustrate that the marginal median 
optimal treatment regime does not ensure  
across-group fairness (defined in Section~\ref{sec:introduction}).
In Section~\ref{sec:motivating-example},
we use a motivating example 
to summarize the differences among the three policies.

\subsection{$d^*_\text{MEAN}$ and $d^*_\text{ACME}$}
\label{sec:d_ate_d_acme}

To characterize the difference between the mean optimal treatment regime $d^*_\text{MEAN}$ and our proposed median optimal treatment regime $d^*_\text{ACME}$, 
we begin by stating a simple condition that ensures the two to be equivalent:   
$ d^*_\text{MEAN}(x) = d^*_\text{ACME}(x)$
if and only if the conditional mean and median treatment effects are always the same sign. 
One sufficient condition for the two policies to be equivalent is thus given below.

\begin{proposition}
If the conditional mean and median treatment effects are separated from zero relative to the conditional outcome standard deviation, in the sense that
$$|\mu_1(x) - \mu_0(x)| > \sigma_1 (x) + \sigma_0 (x) \ \text{ and } \ |m_1(x) - m_0(x)| > \sigma_1 (x) + \sigma_0 (x) , $$ 
then  $d^*_\text{MEAN}(x) = d^*_\text{ACME}(x)$.
\label{prop:ate-acme-variance}
\end{proposition}

An important note about the mean optimal policy $d^*_\text{MEAN}$ is that it does not necessarily satisfy the within-group robustness mentioned in Section~\ref{sec:introduction}, in that it is sensitive to outliers. 
For example, as illustrated in Section~\ref{sec:motivating-example}, 
its decision can be very affected by a small subgroup that benefits a lot from the treatment, even if 
 the majority of the subgroup is harmed by the decision. 
In contrast, the median remains unchanged as long as the amount of mass 
below and above it remains unchanged, 
while the mean is determined by how the masses are distributed, e.g., the mean can be arbitrarily high (or low) by 
moving a small fraction of the mass above (or below) the median 
to extreme values \citep{huber1967behavior, casella2002statistical, huber2004robust}. 
Further comparisons between the mean optimal and our proposed median optimal treatment regime are given in Section~\ref{sec:motivating-example}.

\subsection{$d^*_\text{MME}$ and $d^*_\text{ACME}$}
\label{sec:d_acme_d_mte}
To illustrate ideas, here we consider a simple Gaussian model 
to compare the marginal median optimal treatment regime $d^*_\text{MME}$ and our proposed median optimal treatment regime $d^*_\text{ACME}$. This provides a counterexample showing that $d^*_\text{MME}(x)$ is 
determined not just by the conditional outcome distribution $[Y \mid X=x, A=a]$ for $a \in \{0,1\}$ 
but also by the conditional outcome distribution {at other (different) covariate values}. 
That is, even if the conditional distribution $[Y\mid X=x, A=a]$ 
for $a \in \{0,1\}$ is unchanged, the marginal median optimal policy
$d^*_\text{MME}(x)$ can be different 
depending on $[Y\mid X=x', A=a]$ for $x' \neq x$.  
On the other hand, our median optimal policy $d^*_\text{ACME}(x)$ is fair across groups, in that  the optimal decision 
remains the same so long as 
the conditional distribution for $X=x$ is unchanged.

Specifically consider a data-generating process where
\begin{align*}
X &\sim \text{Bernoulli}(1/2) \\
Y \mid X=x,A=a &\sim N\left( \mu_a(x), \sigma_a^2(x) \right)
\end{align*}
for some unspecified (for now) mean $\mu_a(x)$ and standard deviation $\sigma_a(x)$.  
For any policy $d \in \mathcal{D}$, 
the ACME is 
$
    \mathbb{E}[m_d(X)]= \left(\mu_d(0) + \mu_d(1)\right)/2,
$
and 
the marginal median value %
$m(Y^d)$ satisfies
$
    \Phi\left(\frac{%
    m(Y^d) - \mu_{d}(0)}{\sigma_{d}(0)}\right) 
    + \Phi\left(\frac{%
    m(Y^d) - \mu_{d}(1)}{\sigma_{d} (1)}\right) =1,
$
which implies that %
\begin{align*}
        m(Y^d) = \frac{\sigma_{d}(1)\mu_{d}(0) + \sigma_{d}(0)\mu_{d}(1)}{\sigma_{d}(0) + \sigma_{d}(1)}, 
\end{align*}
i.e., it is a standard deviation-weighted average of outcome regressions (where each regression is weighted by the standard deviation of the other group). In the next proposition we give a counterexample showing how, somewhat surprisingly,  the marginal median optimal policy for subjects with $X=1$, for example, can depend on the means and variances for different subjects with $X=0$. 

\begin{proposition}
Suppose outcomes are higher and more variable under treatment for $X=1$, i.e., 
$$ \mu_1(1) > \mu_0 (1)  \ \text{ and } \  
\sigma_1(1) > \sigma_0(1), $$
and the variances differ more than the means in the sense that  $\frac{\mu_1(1)}{\mu_0(1)} < \frac{\sigma_1(1)}{\sigma_0(1)}$. 
Then there exist conditional distributions $Q_a$ and $\overline{Q}_a$ for $[Y \mid X=0,A=a]$, $a \in \{0,1\}$ such that 
$$ \mathbb{1}\{\mu_1(1) > \mu_0(1)\} = d^*_\text{MME}(1; Q_0, Q_1) \neq d^*_\text{MME}(1; \overline{Q}_0, \overline{Q}_1) = \mathbb{1}\{\mu_1(1) \leq \mu_0(1)\}$$
where $d^*_\text{MME}(x;Q_0, Q_1)$ is the marginal median optimal policy when $[Y \mid X=x,A=a]$ follows distribution $Q_a$.
\label{prop:GoM}
\end{proposition}

\begin{remark}
When $\sigma_1(0) = \sigma_0(0) = \sigma_1(1) = \sigma_0(1)$, 
it follows that  $m(Y^d) = \mathbb{E}[m_d(X)]$ for all $d \in \mathcal{D}$
and 
$d^*_\text{ACME}(x) = \mathbb{1}\{\mu_1(x) > \mu_0 (x)\}$ is an optimal policy with respect to $m(Y^d)$.%
\end{remark}

\begin{remark}[Distribution shift]
Importantly, we also note that among the three optimal policies, 
only the marginal median optimal policy $d^*_\text{MME}$ varies under different covariate distribution $\mathbb{P}(X)$. This is crucial when distribution shift is a concern, which is common in many current setups \citep{quionero2009dataset, mo2020learning}. 
Therefore, under distribution shift 
between training and testing times, 
the marginal median optimal policy $d^*_\text{MME}$ could be suboptimal at testing time even if optimal at training. 
For example, 
if during test time $\mathbb{P}(X = 0) = 0$ and $\mathbb{P}(X = 1) = 1$, 
then 
in the case when $d^*_\text{MME}(1) = \mathbb{1}\{\mu_1(1) \leq \mu_0(1)\}$ during training and $\mu_1(1) \neq \mu_0(1)$, 
$d^*_\text{MME}$ is not optimal (with respect to the marginal median) during test time since $m(Y^d) = \mu_d(1)$. 
Unlike $d^*_\text{MME}$, our proposed median optimal treatment regime   
$d^*_\text{ACME}$ remains the same under covariate shifts. 
\end{remark}

As illustrated in Proposition~\ref{prop:GoM}, 
the optimal decision for a specific subject of covariate $x$
with respect to the marginal median 
depends on not just its own conditional distribution $[Y\mid X=x, A=a]$ for $a \in \{0,1\}$ 
but also the conditional distributions at other covariate values. 
On the other hand, ACME allows the optimal decision for $x$ 
to be only influenced by its own conditional outcome distribution at that $x$. 
In Section~\ref{sec:motivating-example}, 
we give an illustration that summarizes the differences among 
the three policies.

\begin{figure}[H]
    \centering
    \includegraphics[width=\linewidth]{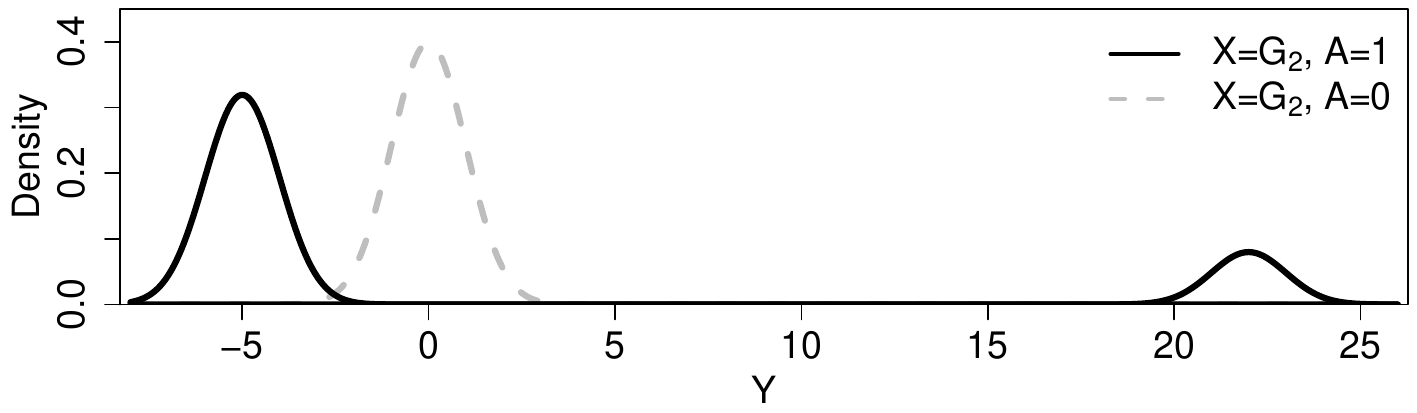}
    \caption{The density of $[Y \mid X={G_2}, A=a]$ in Case I and Case II.}
    \label{fig:female_density}
\end{figure}

\subsection{Simple Illustration Comparing $d^*_\text{MEAN}$, $d^*_\text{MME}$, and $d^*_\text{ACME}$}
\label{sec:motivating-example}

Here we give a simple illustration that showcases the differences among the three policies, and among the corresponding values they maximize. We will see how the ACME  optimal policy $d^*_\text{ACME}$ exhibit both within-group robustness and across-group fairness, while the more standard policies $d^*_\text{MEAN}$ and $d^*_\text{MME}$ may not. 

Consider a population of two groups: one with college degrees ($G_1$) and the other without ($G_2$). 
The binary treatment $A$ indicates whether a job training program is assigned to the individual   
and the  outcome $Y$ is the subjects' improvement in income. In Case I, suppose the groups and outcomes follow the distribution:
\begin{align*}
X &\sim \text{Bernoulli}(1/2) \\
Y \mid X=G_1,A=a &\sim \mathcal{N}\left(9, 1 \right) \\
Y \mid X=G_2,A=a &\sim (1-a) \mathcal{N}(0,1) + a (.2 \; \mathcal{N}(22,1) + .8 \; \mathcal{N}(-5,1))
\end{align*}
In other words, there is no treatment effect for subjects with a college degree, whose outcome distribution is always the same Gaussian.
However, under control subjects without a college degree have  outcomes that follow a Gaussian centered at $0$, while under treatment the outcomes follow a Gaussian mixture. The outcome distribution for $G_2$ is given in Figure~\ref{fig:female_density}. 
When examining the optimal treatment under each value, 
we find that both the mean optimal policy $d^*_\text{MEAN}$ and marginal median optimal policy $d^*_\text{MME}$ treat all subjects without a college degree (Table~\ref{table:illustration}), 
while $d^*_\text{ACME}$ does not treat them. 
Compared to $d^*_\text{MEAN}$, our proposed median optimal policy $d^*_\text{ACME}$ 
is robust against outliers, hence its decisions 
promotes within-group robustness in the sense that 
the decisions are less prone 
to only benefiting a small subgroup, e.g., 
as in our $A=1$ case for the group without a college degree.
On the other hand, although intuitively the  marginal median optimal policy $d^*_\text{MME}$ also gives robust decisions, 
it overlooks characteristics of specific subgroups. 
In particular, treatment decisions for those with $X=x$ can depend on other subjects' conditional outcome distribution with $X \neq x$, violating across-group fairness. 
This is illustrated by comparing the optimal marginal median decision for subjects without a college degree under Case I
and Case II
where 
in Case II
the the conditional outcome distribution for subjects with a college degree under treatment and control is changed to  
\begin{align*}
Y \mid X=G_1, A=a &\sim \mathcal{N}\left(-4, 2 \right).
\end{align*}
Surprisingly, as in Proposition~\ref{prop:GoM}, the marginal median optimal policy $d^*_\text{MME}$ under Case II is different for individuals without a college degree compared to the $d^*_\text{MME}$ under Case I, even though it was only the outcome distribution {for subjects with a college degree} that changed. This illustrates across-group unfairness of the marginal median value. In contrast, the median optimal decision %
for individuals without a college degree has remained unchanged.

\begin{table}[H]
\centering
\begin{tabular}{@{}ccccc@{}}
\toprule
                  & \multicolumn{2}{c}{Case I}             & \multicolumn{2}{c}{Case II}            \\ \midrule
Policy            & Treat $G_1$ & Treat $G_2$              & Treat $G_1$ & Treat $G_2$              \\
$d^*_\text{MEAN}$ & ---        & \checkmark & ---        & \checkmark \\
$d^*_\text{MME}$  & ---        & \checkmark & ---        &                           \\
$d^*_\text{ACME}$ & ---        &                           & ---        &          \\ \bottomrule       
\end{tabular}
\caption{The table shows the treatment assignments under 
the mean optimal policy $d^*_\text{MEAN}$, 
the marginal median optimal policy $d^*_\text{MME}$, 
and the median optimal policy $d^*_\text{ACME}$ in Case I and Case II. 
Since in both cases, the conditional outcome distributions for $G_1$ are the same under treatment and control,
``---'' is used to indicate that the optimal decision for $G_1$ can be either $a=0$ or $a=1$.  
Though the conditional outcome distributions for $G_2$ remain unchanged in both cases, 
the optimal marginal median decision for $G_2$ has changed.
}
\label{table:illustration}
\end{table}

\section{Efficiency Bound}
\label{sec:efficiency}
In this section, using semiparametric theory~\citep{newey1990semiparametric, bickel1993efficient, van2002semiparametric, van2003unified, tsiatis2007semiparametric, kennedy2016semiparametric,diaz2017efficient,chernozhukov2018double},  
we study the nonparametric efficiency bound for 
estimating the ACME $\psi_d := \mathbb{E}[m_d(X)]$ of a given policy $d$, 
given iid samples from distribution $\mathbb{P}$. 
One of the key tools we will use throughout the paper is 
the efficient influence function 
$\phi_d(Z)$ (Corollary~\ref{corollary:influence-function-calculation})
which serves as a first-order derivative 
in a von Mises-type expansion of the target parameter (Lemma~\ref{lemma:decomposition}). 

There are a number of reasons why characterizing efficient influence functions is important.  
The variance of the efficient influence function gives a minimax lower bound for estimating $\psi_d$ (Theorem~\ref{thm:variance})
and so provides a benchmark to compare against when constructing estimators. 
In particular, as used in Section~\ref{sec:estimation}, 
efficient influence functions suggest a doubly robust-style bias-corrected 
estimator.  
Doubly robust estimators attain fast parametric rates 
in nonparametric settings where nuisance functions 
(e.g., $m_a, \pi_a, f_{a,m}$ in our context) 
are estimated at slower rates. 
We begin with presenting the von Mises-type expansion 
(a distributional Taylor expansion)
for ACME, 
which is crucial for the efficiency bound (Theorem~\ref{thm:variance})
and the estimation guarantees for our proposed doubly robust-style estimator (Theorem~\ref{thm:fixed_policy_estimator}).

\begin{lemma}
Given a policy $d \in \mathcal{D}$, 
for the average conditional median effect 
$\psi_d$ defined in~\eqref{eq:ACME}, 
the following decomposition holds for all distributions  $\overline{\mathbb{P}}$ and $\mathbb{P}$:  
\begin{align*}
    \psi_d(\overline{\mathbb{P}}) - \psi_d(\mathbb{P})
    = \int \phi_d(Z; \overline{\mathbb{P}})  d (\overline{\mathbb{P}} - \mathbb{P})
    + R_d(\overline{\mathbb{P}}, {\mathbb{P}}),
\end{align*}
where $\phi_d(Z;\mathbb{P}) = \xi_d(Z;\mathbb{P}) -  \psi_d(\mathbb{P})$, 
\begin{align}
    \xi_d(Z ; \mathbb{P}) &=  \frac{Ad(X)}{\pi_1(X)}\frac{1/2 - \mathbb{1}\{Y \leq m_1(X) \}}{ f_{1,m}\left(X\right)} 
    + \frac{(1-A)(1-d(X))}{\pi_0(X)}\frac{1/2 - \mathbb{1}\{Y \leq m_0(X) \}}{ f_{0,m}\left(X\right)} + m_{d}(X),  
\label{eq:acme_eif}
\end{align}
and 
\begin{align}
    R_d(\overline{\mathbb{P}}, {\mathbb{P}})
    &= \int_\mathcal{X} d(x)\left({\overline{m}_1(x) - m_1(x) - \frac{\pi_1(x)}{\overline{\pi}_1(x)}
    {\frac{F_{1,\overline{m}}(x) - {F}_{1,m}(x)}
    {\overline{f}_{1,\overline{m}}(x)}}}
    \right) \nonumber \\
    &\qquad + (1 - d(x)) \left({\overline{m}_0(x) -m_0(x) - \frac{\pi_0(x)}{\overline{\pi}_0(x)}
    \frac{F_{0, \overline{m}}(x) - {F}_{0,m}(x)}
    {\overline{f}_{0, \overline{m}}(x)}}
    \right)d \mathbb{P}(x).
\label{eq:R_d_raw}
\end{align}
$\overline{\pi}_1, \overline{m}_a$ and $\overline{f}_{a, \overline{m}}$ 
are the propensity score, the conditional median 
and the conditional density at the conditional median 
defined under $\overline{\mathbb{P}}$ 
and $F_{a, \overline{m}}(x) = F_a(\overline{m}_a(x) \mid x)$.  
\label{lemma:decomposition}
\end{lemma}

Proofs of the results in this section 
can be found in Appendix~\ref{appendix:efficiency-proofs}. 
Lemma~\ref{lemma:decomposition} has several important implications. 
It implies that the efficient influence function of ACME is $\phi_d(Z)$ 
as shown in Corollary~\ref{corollary:influence-function-calculation}.
The efficient influence function plays a crucial role in the 
bias term for the plug-in estimator $\psi_d(\widehat{\mathbb{P}})$, 
which we will correct to obtain the doubly robust-style estimator given later on in~\eqref{eq:dr_estimator}. 
The decomposition result in~\eqref{eq:R_d_raw} (along with an alternate expression in~\eqref{eq:R_d_double}) is then used to show the convergence rate of the proposed estimator.

Throughout the rest of the paper, we rely on the following model assumption, which ensures that certain boundedness and mild smoothness conditions hold.

\begin{assumption}\label{assump:statistical_models}
The true distribution $\mathbb{P}$ lies in $\mathcal{P}$ where $\mathcal{P}$ is the  statistical model given by
\edit{
\begin{align*}
    \mathcal{P} := \Big\{\mathbb{P} &\; \Big| \; \exists \epsilon \in (0,1),  \mathbb{P}\left(\epsilon < \pi_1(X) \leq 1 - \epsilon \right) =1, \\
    &\forall a \in \{0,1\}, \exists \; 0 < M_0 \leq M_1,  
    \mathbb{P}\left(M_0 \leq f_{a,m}(X) \leq M_1 \right) =1, \text{ and }\\
    & \forall a \in \{0,1\}, y \in \mathcal{Y}, 
    f_a(y \mid X) \text{ is differentiable and L-Lipschitz continuous in } y \text{ almost surely}
    \Big\}.  
\end{align*}
}
\end{assumption}
The conditions of Assumption~\ref{assump:statistical_models} %
only require some boundedness of the propensity score and conditional outcome density, 
as well as some weak smoothness of the conditional density in $y$. 
Importantly, Assumption~\ref{assump:statistical_models} does not require any smoothness of the propensity score or outcome density in $X$.

As written, the reminder term $R_d(\overline{\mathbb{P}}, {\mathbb{P}})$ in Lemma~\ref{lemma:decomposition} does not appear to be second-order, 
i.e., involving products of differences of $\overline{\mathbb{P}}$ and $\mathbb{P}$. 
However, in the next corollary we show that it is in fact second-order, 
under the model assumption (Assumption~\ref{assump:statistical_models}).

\begin{corollary}\label{corollary:R_d_product}  
For all distributions $\mathbb{P} \in \mathcal{P}$, 
given any distribution $ \overline{\mathbb{P}}$,  
the remainder term %
$R_d(\overline{\mathbb{P}}, \mathbb{P})$ defined in~\eqref{eq:R_d_raw} 
is equivalent to  
\begin{align}
    &R_d(\overline{\mathbb{P}}, \mathbb{P})
    = \mathbb{P}\Bigg\{ 
    \frac{d\left(\overline{m}_1 -m_1 \right)}{\overline{\pi}_1  \overline{f}_{1,\overline{m}}} 
    \left(
    \left(\overline{\pi}_1 - \pi_1 \right)  \overline{f}_{1,\overline{m}} 
    +  
    \left(  \overline{f}_{1,\overline{m}} - f_{1,m}\right) \pi_1
    - (\overline{m}_1 - m_1)\frac{f'_{1,c} \pi_1 }{2}\right) \nonumber \\
    & + \frac{(1 - d)\left(\overline{m}_0 -m_0 \right)}{\overline{\pi}_0  \overline{f}_{0,\overline{m}}} 
    \left(
    (\overline{\pi}_0 - \pi_0)  \overline{f}_{0,\overline{m}} 
    +(\overline{f}_{0,\overline{m}} - f_{0,m}) \pi_0 
    - (\overline{m}_0 - m_0) \frac{f'_{0,c} \pi_0 }{2 }\right) \Bigg\},
\label{eq:R_d_double}
\end{align}
where $f'_{a,c} \leq L$
is the derivative of $f_a(y \mid x)$
at $y = c_a(x)$ 
for some value $c_a(x)$ between $m_a(x)$ and $\overline{m}_a(x)$.\footnote{Since $\mathbb{P} \in \mathcal{P}$, $f'_{a,c}$ exists and is bounded by $L$ almost surely. }
\end{corollary}

The fact that the remainder term is second-order implies that $\phi_d$ from Lemma~\ref{lemma:decomposition} is the efficient influence function, as stated in the next corollary. 

\begin{corollary}
For a given policy $d \in \mathcal{D}$, 
the efficient influence function for the average conditional median effect $\psi_d$ %
is $\phi_d(Z) = \xi_d(Z) - \psi_d$ where $\xi_d$ is defined in~\eqref{eq:acme_eif}. %
\label{corollary:influence-function-calculation}
\end{corollary}

Since the variance of $\phi_d$  serves as a nonparametric efficiency bound, in the local minimax sense \citep{van2002semiparametric}, in the next result we give the exact
form of this variance. This shows what factors drive the statistical difficulty in ACME estimation, and demonstrates the difference in efficiency for ACME versus other value measures.

\begin{theorem}
For a given policy $d \in \mathcal{D}$, 
the nonparametric efficiency bound for estimating $\psi_d$ 
is given by the variance $\sigma_d^2 := \text{Var}\{ \phi_d(Z) \}$ where 
\begin{align}
    \sigma_d^2 = \mathbb{E}\left[ 
    \frac{d(X)}{4\pi_1(X)f_{1,m}^2(X)} 
    + \frac{1 - d(X)}{4 \pi_0(X)f_{0,m}^2(X)} \right]
    + \text{Var}\left( m_d(X) \right).
\label{eq:sigma_d}
\end{align}
\label{thm:variance}
\end{theorem}

Theorem \ref{thm:variance} shows how the efficiency bound for the ACME 
is driven by three main factors:
\begin{itemize}
\item the inverse propensity score $1/\pi_1(X)$,
    \item the inverse conditional density at the conditional median 
$1/f_{a,m}\left( X\right)$, and
    \item the heterogeneity of the conditional median $m_d(X)$, i.e.,   $\text{Var}\left( m_d(X) \right)$.
\end{itemize} 

In contrast, recall the nonparametric efficiency bound for the mean value $\mathbb{E}[Y^d]$~\citep[Theorem 1]{hahn1998role}:
\begin{align*}
    \sigma^2_{d,\text{MEAN}}:= \mathbb{E}\left[ 
    \frac{d(X) \sigma^2_1(X)}{\pi_1(X)} 
    + \frac{(1 - d(X)) \sigma^2_0(X)}{\pi_0(X)} \right] 
    + \text{Var}\left( \mu_d(X) \right),
\end{align*}
instead depends on the heterogeneity of the
conditional outcomes $\sigma^2_a(X)$ and the heterogeneity of the
conditional mean $\text{Var}\left( \mu_d(X) \right)$ (along with the inverse propensity score).

Importantly, when the conditional distribution $[Y \mid X=x, A=a]$ is heavy-tailed, e.g., the outcome variance $\sigma^2_a(X)$ is  very large, then in general the mean value efficiency bound 
$\sigma^2_{d,\text{MEAN}}$  would be severely affected, while the ACME analog $\sigma_d^2$ could still be quite small. 

Next we formalize the local minimax lower bound property of the variance $\sigma_d^2$.

\begin{corollary}\citep[Corollary 2.6]{van2002semiparametric}
For a given policy $d \in \mathcal{D}$ and any estimator $\widehat{\psi}_d$ 
learned using $n$ iid samples from $\mathbb{P}$, 
it follows that
\begin{align*}
    \inf_{\delta > 0} 
    \liminf_{n \to \infty} 
    \sup_{\text{TV}(\overline{\mathbb{P}}, \mathbb{P}) < \delta} 
    n \; \mathbb{E}_{\overline{\mathbb{P}}} \left[\left(\widehat{\psi}_d - \psi_d\left(\overline{\mathbb{P}}\right) \right)^2 \right] 
    \geq \sigma_d^2,
\end{align*}
where $\text{TV}(\overline{\mathbb{P}}, \mathbb{P})$ is the total 
variation distance between $\overline{\mathbb{P}}$ and ${\mathbb{P}}$ 
and $\sigma_d^2$ is defined in~\eqref{eq:sigma_d}. 
\label{corollary:efficiency}
\end{corollary}

Corollary~\ref{corollary:efficiency} directly follows from~\citet[Corollary 2.6]{van2002semiparametric}. 
It shows that without further assumptions, the asymptotic local minimax mean squared error of any estimator scaled by $n$  
can be no smaller than $\sigma_d^2$.  
In Section~\ref{sec:estimation}, under mild conditions, 
we construct doubly robust-style estimators 
that achieve the local minimax lower bound shown in Corollary~\ref{corollary:efficiency}.

\section{Estimation of the ACME}
\label{sec:estimation}
In this section, we propose efficient estimators of the ACME that utilize the 
efficient influence function given in Corollary~\ref{corollary:influence-function-calculation}. 
We start by presenting a simple plug-in estimator, 
which is a sample analogue of $\psi_d$ for a given policy $d \in \mathcal{D}$. 
In Section~\ref{sec:dr_estimator}, we present the doubly robust-style 
estimator %
for evaluating the ACME of a fixed given policy, as well as %
the optimal policy, and which improves on the simple plug-in. 
Specifically, under mild conditions, we show that in both cases the doubly robust-style estimator achieves 
the local minimax lower bound shown in Corollary~\ref{corollary:efficiency}. 
Finally, we present an algorithm for constructing the doubly robust-style estimator (Section~\ref{sec:algorithm}).

Let $Z^{n} = (Z_1, \ldots, Z_n)$ 
and $Z^{n,0} = (Z^0_1, \ldots, Z^0_n)$ denote two independent 
samples. 
For a given policy $d \in \mathcal{D}$, 
the most natural estimator for estimating $\psi_d$ is the plug-in estimator, i.e., 
\begin{align}
    \widehat{\psi}_{d, \text{pi}}
    =\mathbb{P}_n \left\{ d(X)\widehat{m}_1(X) 
    + \left(1 - d(X) \right) \widehat{m}_0(X) \right\},
\label{eq:pi_estimator}
\end{align}
where $\widehat{m}_1$ and $\widehat{m}_0$ are learned using a separate sample $Z^{n,0}$ from the sample $Z^n$ used for taking the sample average $\mathbb{P}_n$. 
While the plug-in estimator is intuitive, 
in general it inherits the slow convergence rate 
from the nuisance functions~\citep{bickel1993efficient,van2002semiparametric}. 
As suggested by the von Mises-type decomposition (Lemma~\ref{lemma:decomposition}), 
a natural estimator that improves upon the plug-in estimator is given in~\eqref{eq:dr_estimator}.

\begin{remark}
For a given set of iid data, 
we can obtain $Z^n$ and $Z^{n,0}$ 
by randomly splitting the data in half. 
In general, to obtain full sample size efficiency, 
one can split the data in folds 
and perform cross-fitting~\citep{bickel1988estimating, robins2008higher, zheng2010asymptotic, chernozhukov2018double}, 
i.e., repeat the learning procedure 
for each split and average the results. 
Throughout this section, 
to ease notation, 
we present the analyses and results for the learning procedure with a single sample split, 
which can be easily extended to averages of multiple independent  splits. If one wants to avoid sample splitting, our results would still hold under appropriate empirical process conditions (e.g., the nuisance functions taking values in Donsker classes).
\end{remark}

\subsection{Doubly Robust-Style Estimator of the ACME}
\label{sec:dr_estimator}

To improve on the potential deficiencies of the simple plug-in above, we propose the doubly robust-style estimator $\widehat{\psi}_{d,\text{dr}}$:  
\begin{align}
\widehat{\psi}_{d,\text{dr}} 
&= %
\mathbb{P}_n \left \{
\frac{\mathbb{1}\{A=d(X)\}}{A\widehat\pi_1(X) + (1-A)\widehat\pi_0(X)}\frac{1/2 - \mathbb{1}\{Y \leq \widehat{m}_A(X) \}}{ \widehat{f}_{A,\widehat{m}}\left(X \right)} 
+ \widehat{m}_d(X)
\right\},
\label{eq:dr_estimator}
\end{align}
where the nuisance functions $\widehat{\pi}_1$, $\widehat{m}_a$ and $\widehat{f}_{a,\widehat{m}}$ are learned using the separate sample $Z^{n,0}$, and $\widehat{\pi}_0(X) = 1 - \widehat{\pi}_1(X)$. 
Importantly, this estimator uses the efficient influence function (Corollary~\ref{corollary:influence-function-calculation}) to correct the bias of the plug-in. 

\begin{remark}
As will be shown in Theorem~\ref{thm:fixed_policy_estimator}, our estimator is not strictly doubly robust in the usual sense (i.e., its consistency requires consistent estimation of the conditional median). Nonetheless, we still call it doubly robust-style because it {does} have what is arguably the most crucial property of doubly robust estimators: its error involves \emph{second-order} products, and so is ``doubly small''.
\end{remark}

In Section~\ref{sec:acme_estimation_fixed_policy}, 
we study the asymptotic convergence rate of $\widehat{\psi}_{d,\text{dr}}$ 
for a given fixed policy $d$, showing it can converge at parametric rates even when nuisance functions are estimated nonparametrically, and follows a straightforward asymptotic normal distribution.
Then in Section~\ref{sec:acme_estimation_learned_policy}, 
we show that under mild conditions, our estimator $\widehat{\psi}_{d, \text{dr}}$ still exhibits $\sqrt{n}$-convergence and asymptotic normality for the optimal ACME value, even after plugging in the learned policy 
$\mathbb{1}\{\widehat{m}_1(X) > \widehat{m}_0(X)\}$, under a margin condition.

\subsubsection{ACME of a Fixed Policy}
\label{sec:acme_estimation_fixed_policy}

We begin with the case when a fixed policy $d$ is given 
and we aim to estimate the ACME of it. This can be useful if we have a specified class of policies, which we aim to optimize over; this  includes the case where we use an independent sample or split to learn a policy, and then condition on that sample when estimating the value using another. 
In such cases, under mild assumptions, 
we show in Theorem~\ref{thm:fixed_policy_estimator}
that the error of $\widehat{\psi}_{d, \text{dr}}$
is the sum of a centered sample average term which is asymptotically normal  
and a term containing products of nuisance estimation errors. 

\begin{theorem}
Assume
\begin{enumerate}
    \item $\mathbb{P}\left(\epsilon \leq \widehat{\pi}_1(X) \leq 1-\epsilon \right) = 1$  
for some $\epsilon \in (0,1)$,
\item $\mathbb{P}( M_0 \leq \widehat{f}_{a,\widehat{m}}(X) \leq M_1 ) = 1$
for $a \in \{0,1\}$ and $0 < M_0 \leq M_1$, and
\item $\|\widehat{\xi}_d - \xi_d\| = o_{\mathbb{P}}(1)$.
\end{enumerate}
Then, 
\begin{align*}
    \widehat{\psi}_{d, \text{dr}} -  \psi_d = 
    {\big(\mathbb{P}_n - \mathbb{P} \big) \xi_d(\mathbb{P})} 
    + O_{\mathbb{P}}\left(  
    \sum_{a=0}^1 \left\|\widehat{m}_a - m_a\right\| \; 
    \left( \|\widehat{\pi}_a \widehat{f}_{a,\widehat{m}}  - \pi_a {f}_{a, m} \| 
    + \|\widehat{m}_a - m_a\| \right) + o_\mathbb{P}(1/\sqrt{n}) \right).
\end{align*}
\label{thm:fixed_policy_estimator}
\end{theorem}

Importantly, Theorem~\ref{thm:fixed_policy_estimator} implies that
$\widehat{\psi}_{d, \text{dr}}$ attains a faster convergence rate
than its nuisance estimators and can be asymptotically normal even when these nuisance functions only satisfy nonparametric sparsity, 
smoothness or other assumptions. We detail this further in the next corollary.

\begin{corollary} 
Given $d \in \mathcal{D}$, under assumptions in Theorem~\ref{thm:fixed_policy_estimator}, 
and the following two conditions:
\begin{enumerate}
    \item $\|\widehat{m}_a - m_a\| = o_{\mathbb{P}}(n^{-1/4})$, and \item For $a \in \{0,1\}$, 
        $
        \|\widehat{\pi}_a \widehat{f}_{a,\widehat{m}} 
        - \pi_a {f}_{a, m} \|=  O_{\mathbb{P}}(n^{-1/4}) $, 
\end{enumerate}
we have that $\widehat{\psi}_{d, \text{dr}}$ is root-n consistent and asymptotically normal with: 
\begin{align*}
    \sqrt{n}(\widehat{\psi}_{d, \text{dr}} - \psi_d)
    \leadsto \mathcal{N} \left(0, \sigma_d^2\right). 
\end{align*}
\label{corollary:fixed-policy-corollary}
\end{corollary}

\begin{remark}
Note that one sufficient condition for  $ \|\widehat{\pi}_a \widehat{f}_{a,\widehat{m}} 
        - \pi_a {f}_{a, m} \|=  O_{\mathbb{P}}(n^{-1/4}) $
is that $\|\widehat{\pi}_a  - \pi_a\| = O_{\mathbb{P}}(n^{-1/4}) $
and 
$\| \widehat{f}_{a,\widehat{m}}  -  {f}_{a,{m}} \| = O_{\mathbb{P}}(n^{-1/4})$ separately.
\end{remark}

 Corollary~\ref{corollary:fixed-policy-corollary} shows that
$\widehat{\psi}_{d, \text{dr}}$ can achieve the nonparametric 
efficiency bound given in Corollary~\ref{corollary:efficiency}. 
For example, when $\pi_a$, $m_a$ and $f_{a,m}$ 
are $d$-dimensional functions in a H\"older smooth class
with smoothness parameter being 
$\alpha$, $\beta$ and $\gamma$ respectively
(i.e., their partial derivatives up to order $\alpha, \beta, \gamma$ exist and are Lipschitz)
and are estimated with squared error 
$n^{-2\alpha/(2\alpha + d)}$, $n^{-2\beta/(2\beta + d)}$
and $n^{-2\gamma/(2\gamma + d)}$, 
then the conditions in Corollary~\ref{corollary:fixed-policy-corollary} will be satisfied when $\alpha, \beta, \gamma \geq d/2$, i.e., when the smoothness is greater than half the dimension. %

\subsubsection{ACME of the Optimal Policy}
\label{sec:acme_estimation_learned_policy}
\edit{
Instead of evaluating a fixed policy $d$, 
one  may want to evaluate the ACME of the truly
optimal (but unknown) policy $d^*_{\text{ACME}} = \mathbb{1}\{\gamma(X) > 0\}$ where $\gamma(X) := m_1(X) - m_0(X)$ is the conditional median treatment effect (CMTE).  
This would give a benchmark for the best possible value that could be achieved by any policy; any improvements would have to come by way of changing the population or changing the treatment.
A natural estimator for the optimal policy would simply plug in an estimate of $\gamma$: 
$\dacmehat(X) = \mathbb{1}\{\widehat{\gamma}(X) > 0\}$. 
Here, the policy $\dacmehat$ is learned 
through the same sample $Z^{n,0}$ as the one used for estimating 
the nuisance functions~\eqref{eq:dr_estimator}. 
For estimating the CMTE $\gamma$, 
in addition to the plug-in $\widehat{m}_1 - \widehat{m}_0$, 
a doubly robust-style estimator
will be discussed in Section~\ref{sec:learning-gamma}. 
An immediate question to ask is whether the %
estimator
$\widehat{\psi}_{\dacmehat, \text{dr}}$ for the learned 
policy $\dacmehat$
can still exhibit similar asymptotic convergence guarantees 
as $\widehat{\psi}_{d, \dr}$ for a fixed policy $d$.
Our goal in this subsection is to answer this question.
To simplify notation, 
we %
use ${\psi}_{{d}^*} := {\psi}_{\dacme}$ and 
$\widehat{\psi}_{\widehat{d}^*, \text{dr}} := \widehat{\psi}_{\dacmehat, \text{dr}}$. 
}

Importantly, the ACME of the optimal policy is a non-smooth functional, because of its dependence on the indicator function $\mathbb{1}\{\gamma(X) > 0\}$. 
This challenge also arises in the mean value setting \citep{ chakraborty2010inference,laber2011adaptive, hirano2012impossibility,van2014targeted, laber2014dynamic,luedtke2016statistical}. One solution is to incorporate a margin condition \citep{tsybakov2004optimal, luedtke2016statistical}, as follows.

\begin{assumption}[Margin Condition]
For some $\alpha > 0$ and all $t >0$, 
we have 
\begin{align*}
    \mathbb{P}\left(|\gamma(X)| \leq t \right) %
    \leq (ct)^{\alpha},
\end{align*}
for some constant $c >0$ such that $ct \leq 1$. 
\label{assump:margin}
\end{assumption}

The exponent $\alpha$ in the margin condition
characterizes the mass such that  
$m_1(X)$ and $m_0(X)$ are close. 
The lower $\alpha$ is, the weaker the margin condition is, i.e., 
the more mass the distribution of $\gamma(X)$ is allowed to have near zero.  
Similar margin conditions are used in classification literature~\citep{tsybakov2004optimal} 
and optimal treatment regimes~\citep{luedtke2016statistical}. Note that $\alpha=0$ encodes no assumption, allowing $\gamma(X)=0$ almost surely, while $\alpha=1$ would hold as long as $\gamma(X)$ is continuously distributed with bounded density. %
The margin condition therefore provides a characterization on how hard 
the optimal decision problem is---intuitively, when $\gamma(X)$ is near zero, it is very hard to distinguish which subjects will benefit from treatment, while when $\gamma(X)$ is very different from zero, this is easy to distinguish. 
Under the margin condition,
we show in Theorem~\ref{thm:learned_policy_estimator}
that the error of $\widehat{\psi}_{\widehat{d}^*, \text{dr}}$
is similar to that of the error of $\widehat{\psi}_{d^*, \text{dr}}$
for the truly optimal policy $d^*$.

\begin{theorem}\label{thm:learned_policy_estimator}
Assume
\begin{enumerate}
    \item $\mathbb{P}\left(\epsilon \leq \widehat{\pi}_1(X) \leq 1-\epsilon \right) = 1$  
for some $\epsilon \in (0,1)$,
\item $\mathbb{P}( M_0 \leq \widehat{f}_{a,\widehat{m}}(X) \leq M_1 ) = 1$
for $a \in \{0,1\}$ and $0 < M_0 \leq M_1$, and
\item $\|\widehat{\xi}_{\widehat{d}^*} - \xi_{d^*}\| = o_{\mathbb{P}}(1)$.
\end{enumerate}
Then, under Assumption~\ref{assump:margin}, 
we have that 
\begin{align*}
    \widehat{\psi}_{\widehat{d}^*, \text{dr}} -  \psi_{d^*} &= 
    \left(\mathbb{P}_n 
        - \mathbb{P}\right) {\xi}_{{d}^*} \\
    &+ O_{\mathbb{P}}\bigg(
    \sum_{a=0}^1 \left\|\widehat{m}_a - m_a\right\| \; 
    \left( 
    \|\widehat{\pi}_a \widehat{f}_{a,\widehat{m}} -\pi_a {f}_{a, m} \| 
    + \|\widehat{m}_a - m_a\| 
    \right) 
    + \|\widehat{\gamma} - \gamma\|^{1+\alpha}_\infty +
    o_\mathbb{P}(1/\sqrt{n}) \bigg).
\end{align*}
\end{theorem}

Unlike the error of $\widehat{\psi}_{{d}, \text{dr}}$
for a fixed policy $d$ presented in Theorem~\ref{thm:fixed_policy_estimator},
the error of $\widehat{\psi}_{\widehat{d}^*, \text{dr}}$ depends on $\|\widehat{\gamma}-\gamma\|_\infty^{1+\alpha}$. 
When $\alpha$ gets higher, 
the error of $\widehat{\gamma}$ plays less of a role, 
as captured in Corollary~\ref{corollary:learned-policy-corollary}.
Under the conditions in Corollary~\ref{corollary:learned-policy-corollary}
we obtain asymptotic normality of the estimator
and thus can construct asymptotic confidence intervals
as illustrated in our experiments in Section~\ref{sec:experiment_app}.

\begin{corollary} 
Under the assumptions of Theorem~\ref{thm:learned_policy_estimator}, 
and the following three conditions:
\begin{enumerate}
    \item $\|\widehat{m}_a - m_a\| = o_{\mathbb{P}}(n^{-1/4})$,  
    \item $\|\widehat{\gamma} - \gamma\|_\infty = o_{\mathbb{P}}\left(n^{-1/(2(1+\alpha))}\right)$, 
    \item For $a \in \{0,1\}$, 
        $
        \|\widehat{\pi}_a \widehat{f}_{a,\widehat{m}} 
        - \pi_a {f}_{a, m} \|=  O_{\mathbb{P}}(n^{-1/4}) $, 
\end{enumerate}
we have that $\widehat{\psi}_{\widehat{d}^*, \text{dr}}$ is root-n consistent and asymptotically normal with: 
\begin{align*}
    \sqrt{n}(\widehat{\psi}_{\widehat{d}^*, \text{dr}} - \psi_{d^*})
    \leadsto \mathcal{N} \left(0, \sigma_{d^*}^2\right).
\end{align*}
\label{corollary:learned-policy-corollary}
\end{corollary}

As shown in Corollary~\ref{corollary:learned-policy-corollary}, 
the error of $\widehat{\psi}_{\widehat{d}^*, \text{dr}}$ 
shares the same convergence guarantee as 
$\widehat{\psi}_{{d}^*, \text{dr}}$ under the extra assumption 
that the margin condition holds and the estimation error of $\widehat{\gamma}$ is $o_{\mathbb{P}}\left(n^{-1/(2(1+\alpha))}\right)$. 
This suggests that when the CMTE $\gamma$ can be estimated well, or when the margin condition holds in a strong sense, 
the doubly robust-style estimator when based on the estimated policy $\widehat{d}^*$ behaves as if the true optimal policy $d^*$ was instead plugged in.
In Section~\ref{sec:algorithm}, 
we give an algorithm that describes the estimation procedure for our proposed estimator.

\subsection{Construction of the Estimators}
\label{sec:algorithm}
The construction of the estimator contains two main steps: nuisance training and 
policy evaluation (with cross-fitting an additional possible step). 
Let $(D_1, D_2, D_3)$ denote three independent samples of $n$ observations of $(X_i, A_i, Y_i)$.
\begin{enumerate}
    \item[Step 1] Nuisance training:
    \begin{enumerate}
        \item Use $D_{1}$ to construct propensity score estimates $\widehat{\pi}_1$. 
        \item Use $D_{1}$ to construct conditional median estimates $\widehat{m}_a$ for $a \in \{0,1\}$, for example using quantile regression.
        \item Use $D_{2}$ to construct conditional density estimates at the estimated conditional median 
        $\widehat{f}_{a,\widehat{m}}$, 
        for example by regressing 
        $\frac{1}{h}K\left(\frac{Y-\widehat{m}_a(X)}{h} \right)$
        on $X$ among those with $A=a$, for $K$ a standard kernel function. %
    \end{enumerate}
    \item[Step 2]  %
    Evaluate the value of the given fixed policy $d$ or the learned optimal policy 
    $\widehat{d}^*(X) = \mathbb{1}\{\widehat{\gamma}(X) > 0\}$ on $D_3$ through the estimator presented in~\eqref{eq:dr_estimator}. 
    \item[Step 3] Cross-fitting (Optional): 
    Repeat Step 1 and 2 two times by using the dataset in the order $(D_1, D_3, D_2)$, and $(D_2, D_3, D_1)$.
    Use the average of the resulting estimators as the final estimate of the ACME of the evaluated policy. We focus on three folds for simplicity, but alternatives using $k>3$ folds are also possible.
\end{enumerate}

In Section~\ref{sec:experiment_app}, 
we use an example to illustrate this estimation procedure. In the above we suggested quantile regression and a particular conditional density estimation procedure, but our convergence rate results  show that any generic method could be used, as long as it satisfied the high-level $L_2$ error conditions listed there. 

\begin{remark}
In the above we use sample splitting to avoid empirical process conditions. Specifically, we split $D_1$ from $D_2$ in order to avoid conditions when estimating the conditional density $f_{a,m}$, and we split $D_3$ in order to avoid conditions in doing bias correction using the estimated influence function. This splitting could be omitted for if Donsker-type conditions are deemed acceptable. 
\end{remark}

\section{Estimation of the CMTE \& Policy Learning}
\label{sec:learning-gamma}

In this section, we propose efficient estimators for the conditional median treatment effect $\gamma(X) = m_1(X) - m_0(X)$, 
and use it to construct the median optimal treatment regime. 
Such an estimator is of independent interest, 
as it suggests ways to characterize the heterogeneity of the median treatment effect. 
After providing the doubly robust-style estimator $\wh \gamma_\dr$ in Section~\ref{sec:doubly-robust-CME}, 
we show how its estimation error is connected to the performance of the corresponding median optimal policy estimator $\wh d^*_\dr(X) = \mathbb{1}\{\wh \gamma_\dr(X) > 0\}$ (Section~\ref{sec:median-otr-construction}). 
Unlike existing work on estimating the marginal and conditional quantile treatment effect~\citep{chernozhukov2005iv, diaz2017efficient,firpo2007efficient,fortin2011decomposition,frolich2013unconditional,machado2005counterfactual,melly2005decomposition,rothe2010nonparametric},
our proposed nonparametric estimator for the CMTE relies on pseudo-outcome regression.

\subsection{Doubly Robust-Style Estimator of the CMTE}\label{sec:doubly-robust-CME}
Following the conditional average treatment effect estimation procedure proposed in~\citet{kennedy2020optimal},
we construct the doubly robust-style estimator $\wh \gamma_\dr$ as follows: 
Let $(D_1, D_2, D_3)$ denote three independent samples of $n$ observations of $Z_i = (X_i, A_i, Y_i)$.

\begin{enumerate}
    \item[Step 1] Nuisance training: 
    Same as Step 1 in Section~\ref{sec:algorithm}.
    \item[Step 2]  Pseudo-outcome regression:
    Use $D_3$ to construct the pseudo-outcome
    \begin{align*}
        \widehat{g}(Z) = \wh m_1(X) - \wh m_0(X) + \frac{A-\wh \pi_1(X)}{\wh \pi_1(X) \wh \pi_0(X)} \frac{\frac{1}{2} - \mathbb{1}\{Y \leq \wh m_A(X)\}}{\wh f_{A, \wh m}(X)},
    \end{align*}
    and regress it on covariates $X$ 
    to obtain $\wh \gamma_\dr$. The regression estimator $ \widehat{\mathbb{E}}_n$ is given by 
    \begin{align}\label{eq:linear-smoother}
        \wh \gamma_\dr(x) = \widehat{\mathbb{E}}_n\{\wh g(Z)|X=x\} = \sum_{i=1}^n w_i(x; X^{n})\wh g(Z_i),
    \end{align}
    where the weights $w_i(x; X^n)$ are learned using $X^n$ in sample $D_3$. 
    Examples of such linear smoothers $\wh{\E}_n$ include kernel estimators, 
    linear, 
    ridge, local polynomial and RKHS regression, 
    some random forests (e.g., Mondrian and kernel forests), 
    as well as weighted combinations of aforementioned methods~\citep{wasserman2006all}. %
    \item[Step 3] Cross-fitting (Optional): 
    Repeat Step 1 and 2 two times by using the dataset in the order $(D_1, D_3, D_2)$, and $(D_2, D_3, D_1)$.
    Use the average of the resulting estimators as the final $\wh \gamma_\dr$. 
\end{enumerate}

In a recent line of work~\citep{nie2021quasi,kennedy2020optimal},
the conditional treatment effect estimators are commonly compared with an oracle estimator $\wt \gamma$ that has access to the true nuisance functions. 
We define the oracle estimator in our setting as follows. 
\begin{definition}[Oracle]
Denote $g(Z)$ to be 
the pseudo-outcome that depends on the true nuisance fucntions: %
\begin{align*}
    g(Z) = m_1(X) - m_0(X) + \frac{A - \pi_1(X)}{\pi_1(X) \pi_0(X)} \frac{1/2 - \mathbb{1}\{Y \leq m_A(X)\}}{f_{A, m}(X)}.
\end{align*} 
Given independent and identically distributed samples $\{X_i, O_i\}_{i=1}^n$
where $O_i = \gamma(X_i) + \varepsilon_i$
and $\varepsilon_i$ is a 
mean-zero noise defined to be $\varepsilon_i = g(Z_i) - \gamma(X_i)$, 
the oracle is given by 
$$\wt \gamma(x) = \wh{\mathbb{E}}_n\{O|X=x \}.$$
\end{definition}
In other words,  
the oracle $\wt \gamma$ regresses $O$ %
on the covariates $X$ 
using the same linear smoother $\wh{\mathbb{E}}_n$ as the one used in $\wh \gamma_\dr$~\eqref{eq:linear-smoother}. %
The performance of the oracle depends directly on the complexity 
(e.g., smoothness) of $\gamma$ itself,
since the outcome $O$ is the sum of 
$\gamma(X)$ and a mean-zero noise.

We illustrate in the following theorem that
 the mean squared error of 
$\wh \gamma_\dr$ can be upper bounded 
by the oracle error incurred by $\wt \gamma$ 
and products of nuisance errors. 
This allows the CMTE to be estimated 
at a faster rate even when the nuisance estimates 
are obtained at slower rates.

\begin{theorem}\label{thm:gamma-error}
Let $\wh \gamma_\dr$ and $\wt \gamma$ 
denote the conditional median treatment effect estimator 
and oracle defined as above.
Define the weighted norm $\| \cdot \|_w$ and squared-weighted norm $\| \cdot \|_{w^2}$
by 
\begin{align*}
    \|v\|^2_w &= \|v(x)\|^2_w = \sum_{i=1}^n {\frac{|w_i(x; X^n)|}{\sum_j |w_j(x; X^n)|} \int |v(z)|^2 d \mathbb{P}(z|X_i)},
    \\
    \|v\|^2_{w^2} &=\|v(x)\|^2_{w^2} = \sum_{i=1}^n {\frac{|w_i(x; X^n)|^2}{\sum_j |w_j(x; X^n)|^2} \int |v(z)|^2 d \mathbb{P}(z|X_i)}.
\end{align*}
Assume
\begin{enumerate}
    \item $\mathbb{P}\left(\epsilon \leq \widehat{\pi}_1(X) \leq 1-\epsilon \right) = 1$  
for some $\epsilon \in (0,1)$,
\item $\mathbb{P}( M_0 \leq \widehat{f}_{a,\widehat{m}}(X) \leq M_1 ) = 1$
for $a \in \{0,1\}$ and $0 < M_0 \leq M_1$.
\item $\text{Var}{\{g(Z)|X=x\}} \geq \sigma_{\min}^2$ for all $x \in \mathcal{X}$. 
\end{enumerate}
Then, for all $x \in \mathcal{X}$, we have
\begin{align}\label{eq:cmte-bound}
    (\wh \gamma_\dr(x) - \gamma(x))^2
    \lesssim {(\wt \gamma(x) - \gamma(x))^2}
    + b(x) + O_{\mathbb{P}}\left(\|\wh g - g\|^2_{w^2} {\E[(\wt \gamma(x) - \gamma(x))^2]}\right),
\end{align}
where   
$b(x) =  (\sum_{i=1}^n |w_i(x; X^n)|)^2   (\sum_{a=0}^1 \|\wh m_a - m_a\|^2_w (\|\wh \pi_a \wh f_{a, \wh m} - \pi_a f_{a,m}\|_w  + \|\wh m_a - m_a\|_w ) )^2$. 
\end{theorem}

\begin{remark}
    When the pseudo-outcome estimator $\wh g$
    is consistent in the squared-weighted norm, 
    i.e., $\|\wh g - g\|_{w^2} = o_{\mathbb{P}}(1)$,
    we have the third term in~\eqref{eq:cmte-bound}
    to be $o_\mathbb{P}({\E[(\wt \gamma(x) - \gamma(x))^2]})$. 
\end{remark} 

The assumptions of Theorem~\ref{thm:gamma-error} require $g(Z)$ to have variation for all $X=x$. 
Our %
upper bound on the squared error of $\wh \gamma_\dr$ 
contains three terms: 
the first and last terms depend on 
the error of the oracle $\wt \gamma$,
while the middle term 
relies on products of nuisance errors 
with respect to the weighted norm. 
When $\sum_j |w_j(x; X^n)| = 1$, 
the weighted norm $\|\cdot\|_w$ 
weighs the nuisance errors by $w_i(x;X^n)$ 
in a similar way to 
how the linear smoother $\wh \E_n$ weighs the pseudo-outcomes. 
As an example, when the linear smoother $\wh \E_n$ (e.g., nearest neighbor estimator) relies only on local information, 
these weighted norms ensure that the error at $x$ is 
also weighed by only its local information. 
In the more general case, 
the weights for the weighted norms are normalized so that they sum up to $1$. 
As shown in~\citet{stone1977consistent,gyorfi2002distribution}, 
$\sum_{i=1}^n |w_i (x; X^n)| = O_\mathbb{P}(1)$
is a sufficient condition to ensure 
$\wh \E_n$ to be weakly universally consistent. 
In such cases, 
if $\|\wh g - g\|_{w^2} = o_{\mathbb{P}}(1)$,
then the squared error of $\wh \gamma_\dr$
deviates from the error of the oracle $\wt \gamma$
by products of nuisance errors, 
suggesting that the CMTE can be estimated 
at a faster rate compared to the nuisance functions.  
For a more detailed discussion on 
sufficient and necessary conditions 
on $w_i(x; X^n)$ for ensuring
the consistency of $\wh \E_n$, 
we refer the readers to~\citep{stone1977consistent}.

\subsection{Policy Learning}
\label{sec:median-otr-construction}

We briefly touch on the problem of learning the median optimal treatment regime.
There are many approaches for policy learning. 
Among them, the two canonical ways are 
(1) empirical value maximization where the goal is to directly find 
the optimal policy through maximizing the empirical value of policies in a fixed policy class~\citep{zhao2012estimating, zhang2013robust, athey2021policy};
and (2) estimating the conditional treatment effect $\gamma(X)$ first 
and then plugging it into the closed form of the optimal treatment regime ${d}^*(X) = \mathbb{1}\{{\gamma}(X) > 0\}$~\citep{murphy2003optimal, robins2004optimal, van2014targeted}. 
The second approach is closely related to plug-in classifiers~\citep{audibert2007fast}, 
since deciding on the optimal treatment can be viewed as 
a binary classification task. 
\edit{
Another related line of work is robust policy learning~\citep{xiao2019robust, zhang2021robust}. 
Using a model-based approach, 
\citet{xiao2019robust} learns conditional quantile treatment regimes through robust regression. 
On the other hand, 
\citet{zhang2021robust} uses a 
heuristic two-stage nonparametric approach 
for robust policy learning. 
}
For a more comprehensive review on policy learning, 
we refer the readers to~\citet{athey2021policy}.

Using the doubly robust-style estimator $\wh \gamma_\dr$ presented in~\eqref{eq:linear-smoother},
we adopt the second approach and 
construct the median optimal treatment regime through 
$\wh{d}_\dr^*(X) = \mathbb{1}\{\wh \gamma_\dr (X) > 0\}$.
We notice that the error of the policy $\wh{d}_\dr^*$ is bounded by the error of $\wh \gamma_\dr$. 
Under the assumption that 
for all $x \in \mathcal{X}$, 
either $|\gamma(x)| > \delta$ for some $\delta >0$ 
or $\gamma(x)=0$, 
we obtain that 
\begin{align*}%
    \E[\mathbb{1}\{\wh d_\dr^*(x) \neq  d^*(x)\}] 
    \leq \mathbb{P}(|\gamma(x)| \leq |\wh \gamma_\dr(x) - \gamma(x)|)
    \lesssim \mathbb{E}[|\wh \gamma_\dr(x) - \gamma(x)|]\leq \sqrt{\E[(\wh \gamma_\dr(x) - \gamma(x))^2]}. %
\end{align*} 
The first inequality follows from Lemma~\ref{lemma:policy_diff}
and the second inequality holds since 
if $\gamma(x) = 0$, then $d^*(x)$ can be either $1$ or $0$, suggesting that $\wh d_\dr^*(x)$ will always be optimal.
The error of the learned policy $\wh{d}_\dr^*$
can be upper bounded by the square root of the 
expected squared error given in Theorem~\ref{thm:gamma-error}.
The assumption used here is stronger than the margin condition. %
It is of future interest to relax such assumption
and study the performance of $\wh{d}_\dr^*$ in 
more flexible settings.

\section{Experiments}
\label{sec:experiment}

\subsection{Numerical Simulation}
\label{sec:simulation}
To explore the finite-sample properties of the estimator, 
we simulate from 
the following data generating process: 
\begin{align*}
    X &\sim \mathcal{N}(0,\mathbf{I}_5),\\
    \text{logit}\{\pi_1(X)\} &= X^\top \beta
    \text{ where } \beta = (.2, .2, .2, .2, .2),\\
    Y|X,A &\sim \text{Lognormal}(X^\top \beta + A, .25).
\end{align*}
\edit{We note that there is heterogeneity of 
the conditional median treatment effects across covariates $X$, i.e.,
$m_1(X) - m_0(X) = \exp(X^\top \beta + 1) -  \exp(X^\top \beta)
= (e-1) \exp(X^\top \beta)$.
}
To inspect the rate of convergence of the estimator in terms of the nuisance estimation error, 
we constructed the nuisance estimators through the following procedure:
$\widehat{\pi}_1(X)= \text{expit}\{\text{logit}(\pi_1(X)) + \epsilon_{1,n} \}$,  
 $\widehat{m}_a(X) = m_a(X) + \epsilon_{2,n}$,  
and $\widehat{f}_{a,\widehat{m}}(X) = f_{a,m}(X) + \epsilon_{3,n}$  
where $\epsilon_{1,n}, \epsilon_{2,n}, \epsilon_{3,n}$ are independent samples drawn from $\mathcal{N}(n^{-\alpha}, n^{-2\alpha})$.
This construction ensures that the root mean square errors of 
$\widehat{\pi}_a$, $\widehat{m}_a$ 
and $\widehat{f}_{a,\widehat{m}}$ are of order
$O(n^{-\alpha})$. 
The policy used for evaluation is  $d(X) = \mathbb{1}\{X_1 > 0\}$. 
Results shown in Figure~\ref{fig:cnt-dr-plugin} are 
averaged over $1000$ rounds.

\begin{figure}
     \centering
     \begin{subfigure}[b]{.32\linewidth}
         \centering
         \includegraphics[width=\linewidth]{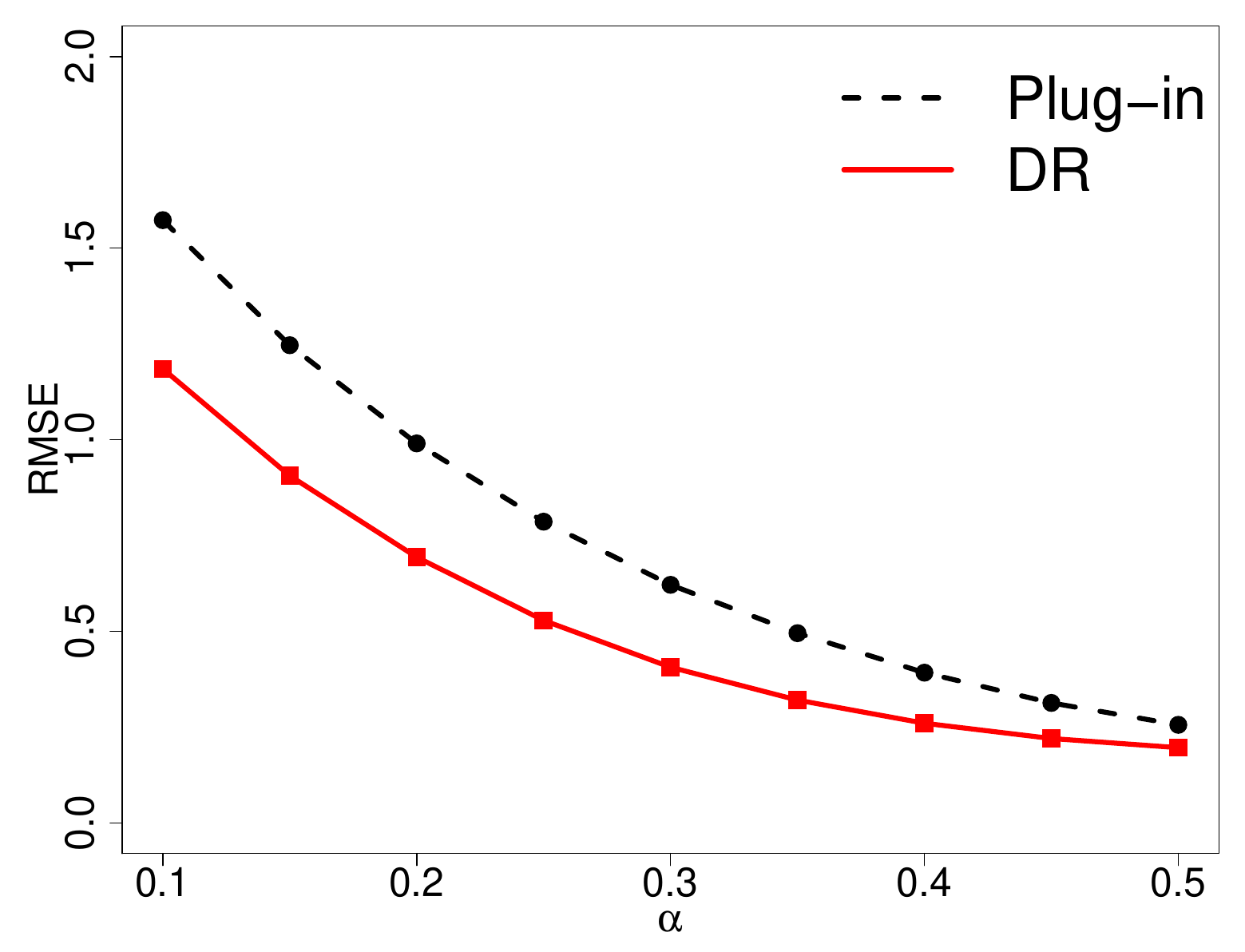}
         \caption{$n=100$}
         \label{fig:n=100}
     \end{subfigure}
     \hfill
     \begin{subfigure}[b]{.32\linewidth}
         \centering
         \includegraphics[width=\linewidth]{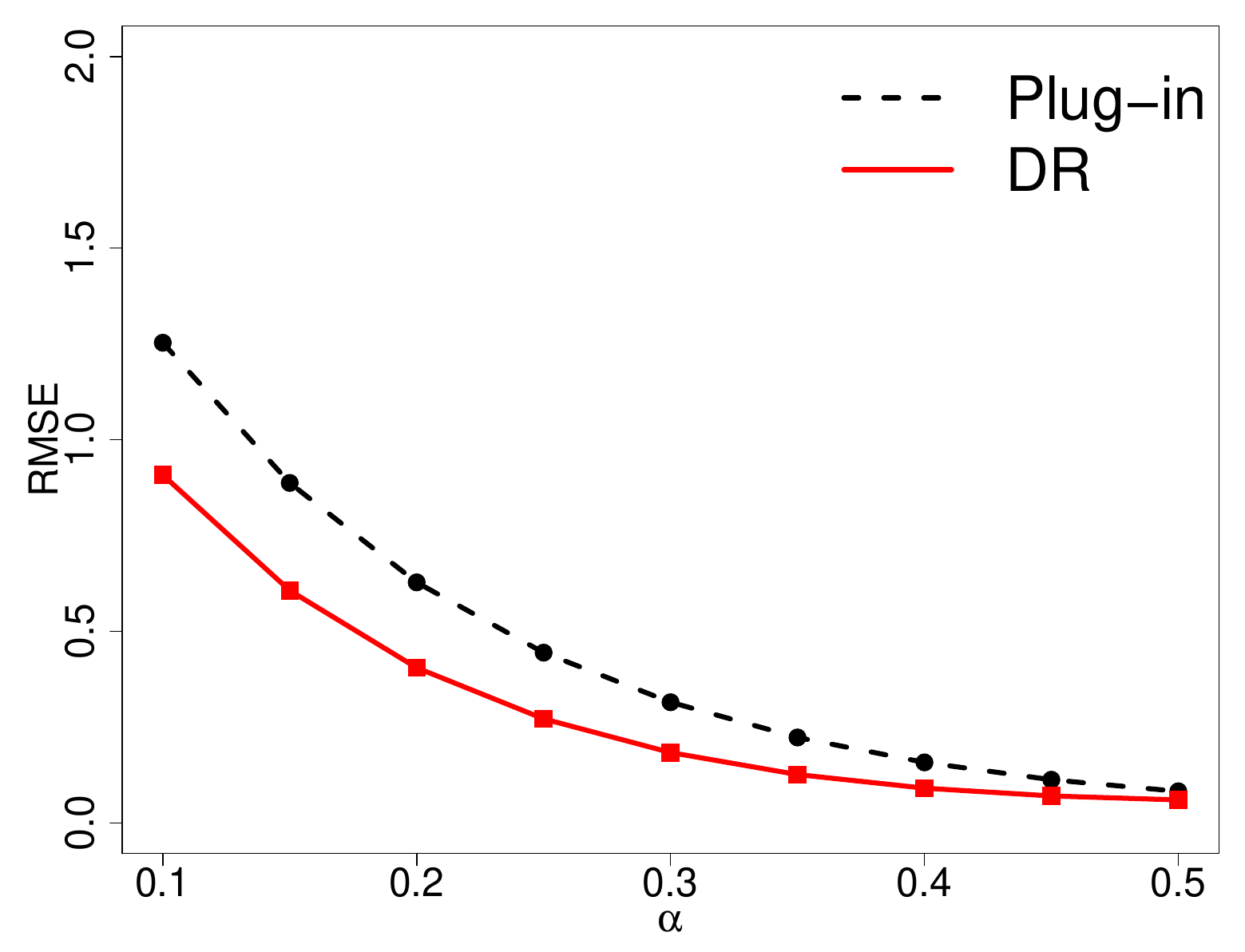}
         \caption{$n=1000$}
         \label{fig:n=1000}
     \end{subfigure}
     \hfill
     \begin{subfigure}[b]{.32\linewidth}
         \centering
         \includegraphics[width=\linewidth]{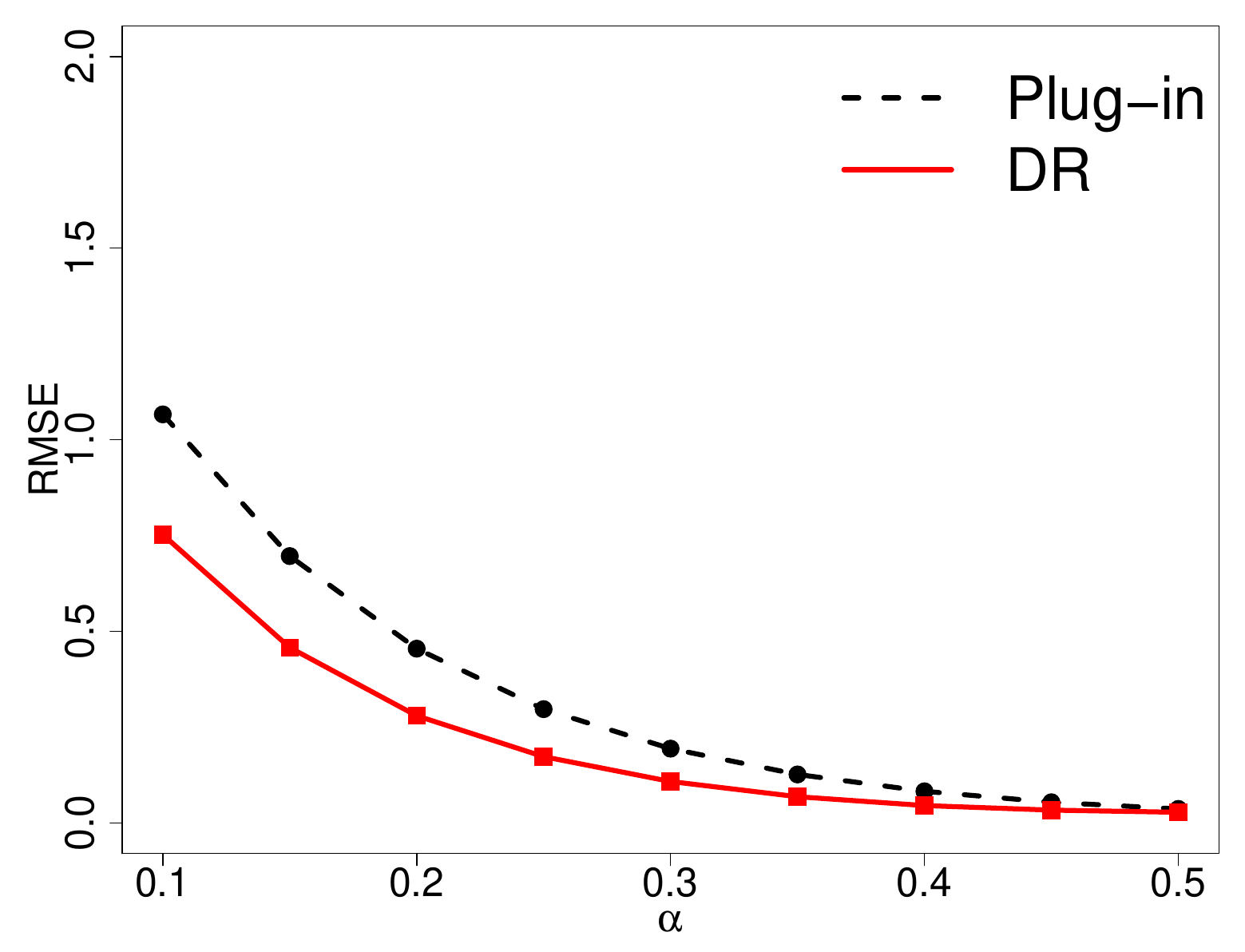}
         \caption{$n=5000$}
         \label{fig:n=5000}
     \end{subfigure}
        \caption{
        Root mean square error of the doubly robust-style estimator~\eqref{eq:dr_estimator} and the plug-in estimator~\eqref{eq:pi_estimator}
        for settings described in Section~\ref{sec:simulation}.
        The sample size $n$ varies in each subplot. 
        Within a subplot, the $x$-axis is $\alpha$, 
        which controls the estimation error of the nuisances. 
        }
        \label{fig:cnt-dr-plugin}
\end{figure}

As we have seen in Figure~\ref{fig:cnt-dr-plugin}, 
the doubly robust-style estimator converges faster compared to the plug-in estimator. 
When $n=5000$, we see the estimation error 
of the doubly robust-style estimator is close to $0$
when $\alpha =.25$, in line with what our theory
suggests in Corollary~\ref{corollary:fixed-policy-corollary}.

\subsection{Application: ACTG 175}
\label{sec:experiment_app}

\begin{figure}
    \centering
    \begin{subfigure}[b]{.48\linewidth}
        \centering
        \includegraphics[width=\linewidth]{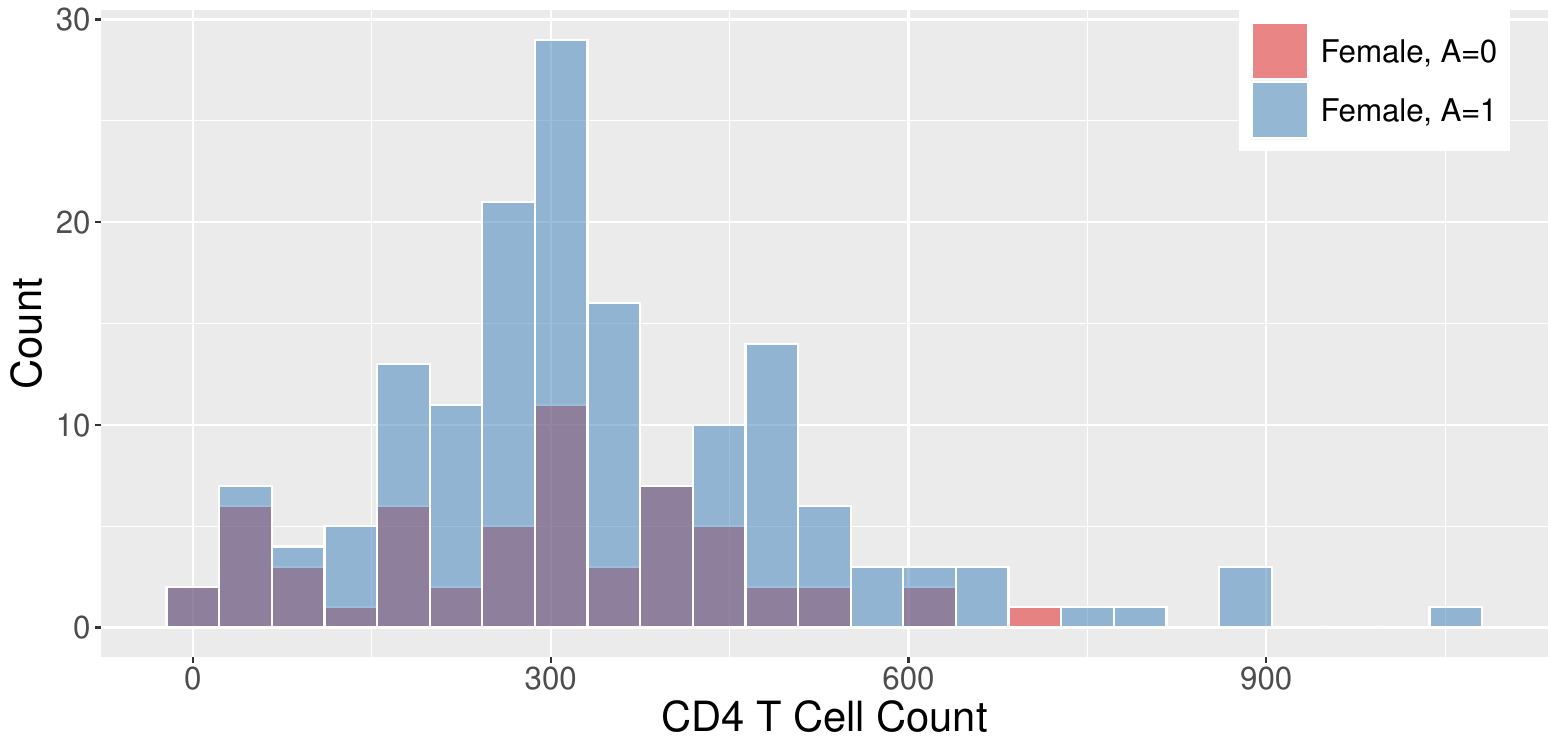}
        \caption{Histogram of the outcomes for females.}
        \label{fig:actg-illustration}
    \end{subfigure}
    \hfill
    \begin{subfigure}[b]{.48\linewidth}
        \centering
        \includegraphics[width=\linewidth]{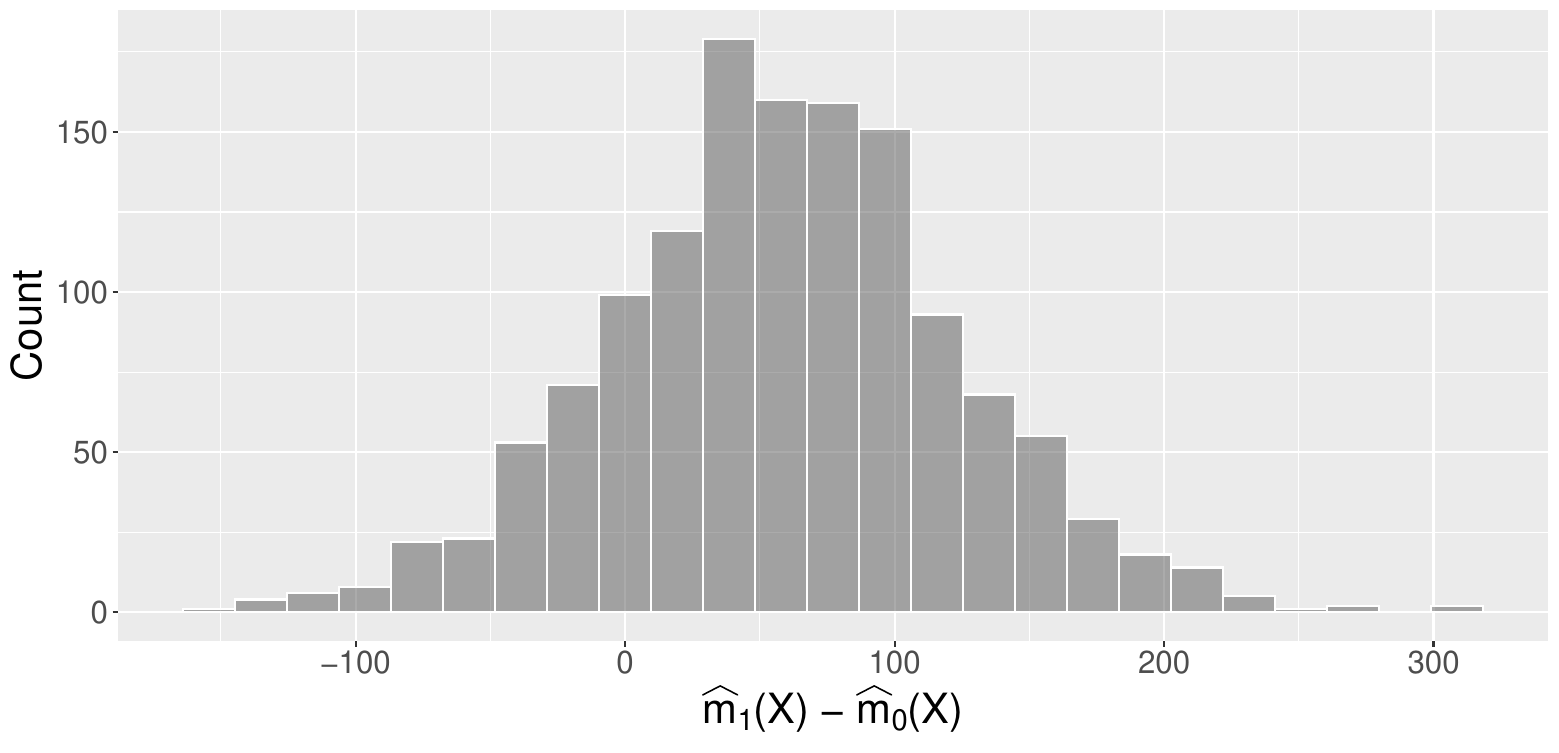}
        \caption{Histogram of $\widehat{m}_1(X_i) - \widehat{m}_0(X_i)$.}
        \label{fig:conditional-median-effect}
    \end{subfigure}
       \caption{
       \ref{fig:actg-illustration} is the histogram 
       of the outcomes for females under treatment and control. 
       The difference in mean outcome under treatment versus control is 49, while the difference in median outcomes is only 1.
       \ref{fig:conditional-median-effect} is the histogram of $\widehat{m}_1(X_i) - \widehat{m}_0(X_i)$ where 
       $\widehat{m}_a$ is estimated using quantile regression forests described in Section~\ref{sec:experiment_app}. 
       }
       \label{fig:summary}
\end{figure}

We illustrate our proposed methods using the ACTG 175 dataset
given in the \texttt{R} package \texttt{speff2trial} \citep{juraska2010package}.
The data are from a randomized clinical trial
in a population of adults with HIV type I. 
The treatment is binary where $A=0$ stands for only using zidovudine as the therapy and $A=1$ represents combination therapies.  
The outcome of interest is CD4 T cell count after
$96 \pm 5$ weeks. 
Covariates $X$ include baseline CD4 and CD8 T cell count, 
age, weight, Karnofsky score, indicators for
race, gender, hemophilia, homosexual activity, drug use, whether symptomatic, and previous
zidovudine and antiretroviral use. 
Since the treatment is randomized, the propensity score is known to be $\pi_1(x)=.75$ for all subjects.
The number of observations is  $n=1342$, after excluding subjects with missing outcomes.  

Figure~\ref{fig:actg-illustration} shows a histogram
of outcomes under treatment versus control for females; the presence of skewed treatment responses suggest the median may be a more useful measure than the mean in this study. 
In fact, although the mean outcome in this group is quite different under treatment versus control (mean CD4 count is 341 for treated females, but only 292 for controls), the medians are similar (313 for treated, 312 for control).

Therefore we apply our proposed methods to estimate the value of the median optimal policy, and the value of a few competing policies. Specifically we use the estimator described in Section~\ref{sec:algorithm}, 
splitting the sample into thirds ($D_1, D_2$ and $D_3$) and using cross-fitting. In $D_1$ the conditional median estimate
$\widehat{m}_a$ is obtained with quantile regression forests via the package~\texttt{quantregForest}~\citep{meinshausen2006quantile} (recall the propensity score is known and so does not need to be estimated here).
$D_2$ is then used to construct the density estimate $\widehat{f}_{a, \widehat{m}}$, 
which we estimated by regressing  a Gaussian kernel-weighted outcome centered at $\widehat{m}_a$ on $X$ using the \texttt{randomForest} R package.  
The bandwidth $h$ was chosen using Silverman's rule~\citep{silverman1986density}. 
We considered estimating the ACME value of five policies using $D_3$:
the observational policy $d(X_i) = A_i$, \edit{a plug-in median optimal policy} $d(X_i) = \mathbb{1}\{\widehat{m}_1(X_i) > \widehat{m}_0(X_i)\}$, the treat-all policy  
$d(X_i)=1$, the treat-none policy $d(X_i)=0$, and \edit{a plug-in mean optimal policy} 
$d(X_i) = \mathbb{1}\{\widehat{\mu}_1(X_i) > \widehat{\mu}_0(X_i)\}$ where the regression functions $\widehat{\mu}_a$ are also estimated via random forests.

\begin{figure}
    \centering
        \includegraphics[width=.7\linewidth]{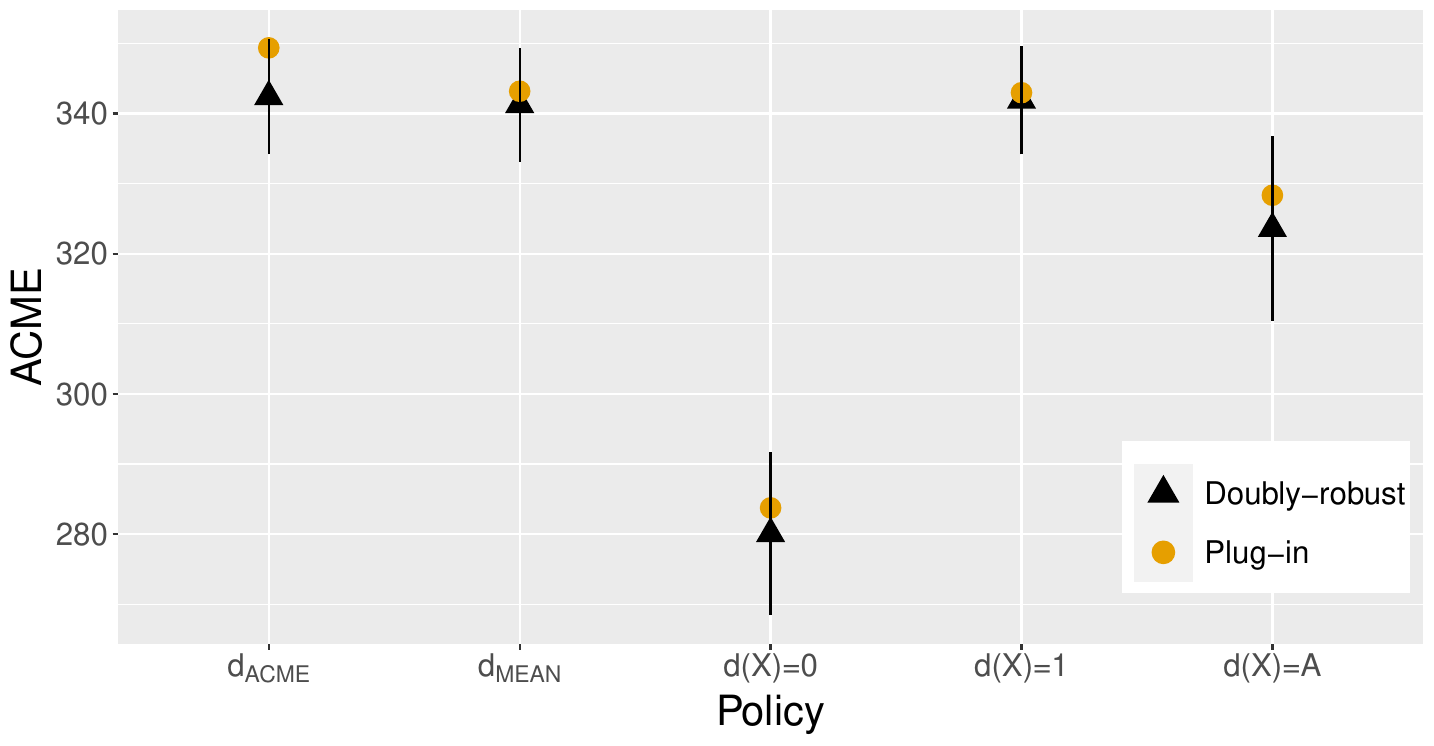}
       \caption{
       The results for the ACTG 175 data analysis. 
       Black triangles show the estimated ACME values of five different policies, computed with our proposed estimator~\eqref{eq:dr_estimator}, and the black lines show the 
       estimated $95\%$ confidence intervals. 
       For comparison, the yellow circles are estimates from the plug-in estimator~\eqref{eq:pi_estimator}. 
       }
       \label{fig:actg}
\end{figure}

Figure~\ref{fig:actg} shows the estimated values with the proposed doubly robust-style estimator $\widehat{\psi}_{d, \text{dr}}$, as well as the plug-in $\widehat{\psi}_{d, \text{pi}}$, for reference. 
 $95\%$ confidence intervals for $\widehat{\psi}_{d, \text{dr}}$
are obtained with the usual Wald interval based on the empirical variance 
$\widehat{\psi}_{d, \text{dr}} \pm 1.96 \widehat{\sigma}_d/\sqrt{n}$
where $\widehat{\sigma}_d$ is the sample standard deviation of the influence function estimates. 
The results show the median optimal policy gives the highest value ($342.40$ with 95\% CI $8.20$), with the mean optimal policy and treat-all policy close behind. The treat-none policy does substantially worse than even the observational (random assignment) policy. %
Figure~\ref{fig:conditional-median-effect} shows a histogram of the conditional median treatment effects $\widehat{m}_1(X_i) - \widehat{m}_0(X_i)$, indicating a modest amount of effect heterogeneity.

\section{Discussion}

In this paper, we proposed a treatment policy based on conditional median treatment effects, and a new ACME value measure which is maximized for this policy. Importantly, our proposed approach avoids both within-group lack of robustness issues with the mean, as well as across-group unfairness issues with the marginal median. We argue that optimal treatment policies should be defined in terms of the \emph{conditional} distribution of outcomes, given covariates, rather than the marginal, in order to avoid across-group unfairness. We study nonparametric efficiency bounds and propose provably optimal doubly robust-style estimators, showing their finite-sample properties in simulations and in an illustration analyzing effects of combination therapy in treating HIV. 
We would generally argue that the mean is most useful as a measure of centrality when it is close to the median, and otherwise the median should be preferred; this suggests median effects should at least be used more widely than they are at present.

There are many opportunities for future related work. In addition to developing $V$-specific median optimal treatment regimes
and analyzing $\wh d_\dr^*$ in more general settings, 
one could also consider more general (conditional) quantile optimal treatment regimes that replace conditional medians by conditional quantiles. All of this could also be adapted to numerous other settings, including continuous or time-varying treatments, or mediation problems, or settings where treatments may be confounded so that sensitivity analysis or instrumental variables may be used, etc.

\section*{Acknowledgement}
LL is generously supported by an Open Philanthropy AI Fellowship, and EK gratefully acknowledges support from NSF Grant DMS1810979.
The authors  thank 
Sivaraman Balakrishnan
and 
David Childers 
for very helpful discussions.

\bibliographystyle{plainnat}
\bibliography{causality-refs}

\begin{thebibliography}{64}
\providecommand{\natexlab}[1]{#1}
\providecommand{\url}[1]{\texttt{#1}}
\expandafter\ifx\csname urlstyle\endcsname\relax
  \providecommand{\doi}[1]{doi: #1}\else
  \providecommand{\doi}{doi: \begingroup \urlstyle{rm}\Url}\fi

\bibitem[Athey and Wager(2021)]{athey2021policy}
Susan Athey and Stefan Wager.
\newblock Policy learning with observational data.
\newblock \emph{Econometrica}, 89\penalty0 (1):\penalty0 133--161, 2021.

\bibitem[Audibert and Tsybakov(2007)]{audibert2007fast}
Jean-Yves Audibert and Alexandre~B Tsybakov.
\newblock Fast learning rates for plug-in classifiers.
\newblock \emph{The Annals of statistics}, 35\penalty0 (2):\penalty0 608--633,
  2007.

\bibitem[Bickel and Ritov(1988)]{bickel1988estimating}
Peter~J Bickel and Yaacov Ritov.
\newblock Estimating integrated squared density derivatives: sharp best order
  of convergence estimates.
\newblock \emph{Sankhy{\=a}: The Indian Journal of Statistics, Series A}, pages
  381--393, 1988.

\bibitem[Bickel et~al.(1993)Bickel, Klaassen, Bickel, Ritov, Klaassen, Wellner,
  and Ritov]{bickel1993efficient}
Peter~J Bickel, Chris~AJ Klaassen, Peter~J Bickel, Ya’acov Ritov, J~Klaassen,
  Jon~A Wellner, and YA'Acov Ritov.
\newblock \emph{Efficient and adaptive estimation for semiparametric models},
  volume~4.
\newblock Johns Hopkins University Press Baltimore, 1993.

\bibitem[Casella and Berger(2002)]{casella2002statistical}
George Casella and Roger~L Berger.
\newblock \emph{Statistical inference}, volume~2.
\newblock Duxbury Pacific Grove, CA, 2002.

\bibitem[Chakraborty(2013)]{chakraborty2013statistical}
Bibhas Chakraborty.
\newblock \emph{Statistical methods for dynamic treatment regimes}.
\newblock Springer, 2013.

\bibitem[Chakraborty et~al.(2010)Chakraborty, Murphy, and
  Strecher]{chakraborty2010inference}
Bibhas Chakraborty, Susan Murphy, and Victor Strecher.
\newblock Inference for non-regular parameters in optimal dynamic treatment
  regimes.
\newblock \emph{Statistical methods in medical research}, 19\penalty0
  (3):\penalty0 317--343, 2010.

\bibitem[Chernozhukov and Hansen(2005)]{chernozhukov2005iv}
Victor Chernozhukov and Christian Hansen.
\newblock An iv model of quantile treatment effects.
\newblock \emph{Econometrica}, 73\penalty0 (1):\penalty0 245--261, 2005.

\bibitem[Chernozhukov et~al.(2018)Chernozhukov, Chetverikov, Demirer, Duflo,
  Hansen, Newey, and Robins]{chernozhukov2018double}
Victor Chernozhukov, Denis Chetverikov, Mert Demirer, Esther Duflo, Christian
  Hansen, Whitney Newey, and James Robins.
\newblock Double/debiased machine learning for treatment and structural
  parameters.
\newblock \emph{The Econometrics Journal}, 21\penalty0 (1):\penalty0 C1--C68,
  2018.

\bibitem[Chouldechova and Roth(2020)]{chouldechova2018frontiers}
Alexandra Chouldechova and Aaron Roth.
\newblock A snapshot of the frontiers of fairness in machine learning.
\newblock \emph{Communications of the ACM}, 63\penalty0 (5):\penalty0 82--89,
  2020.

\bibitem[D{\'\i}az(2017)]{diaz2017efficient}
Iv{\'a}n D{\'\i}az.
\newblock Efficient estimation of quantiles in missing data models.
\newblock \emph{Journal of Statistical Planning and Inference}, 190:\penalty0
  39--51, 2017.

\bibitem[Dwork et~al.(2012)Dwork, Hardt, Pitassi, Reingold, and
  Zemel]{dwork2012fairness}
Cynthia Dwork, Moritz Hardt, Toniann Pitassi, Omer Reingold, and Richard Zemel.
\newblock Fairness through awareness.
\newblock In \emph{Proceedings of the 3rd innovations in theoretical computer
  science conference}, pages 214--226, 2012.

\bibitem[Firpo(2007)]{firpo2007efficient}
Sergio Firpo.
\newblock Efficient semiparametric estimation of quantile treatment effects.
\newblock \emph{Econometrica}, 75\penalty0 (1):\penalty0 259--276, 2007.

\bibitem[Fortin et~al.(2011)Fortin, Lemieux, and
  Firpo]{fortin2011decomposition}
Nicole Fortin, Thomas Lemieux, and Sergio Firpo.
\newblock Decomposition methods in economics.
\newblock In \emph{Handbook of labor economics}, volume~4, pages 1--102.
  Elsevier, 2011.

\bibitem[Fr{\"o}lich and Melly(2013)]{frolich2013unconditional}
Markus Fr{\"o}lich and Blaise Melly.
\newblock Unconditional quantile treatment effects under endogeneity.
\newblock \emph{Journal of Business \& Economic Statistics}, 31\penalty0
  (3):\penalty0 346--357, 2013.

\bibitem[Gy{\"o}rfi et~al.(2002)Gy{\"o}rfi, Kohler, Krzyzak, Walk,
  et~al.]{gyorfi2002distribution}
L{\'a}szl{\'o} Gy{\"o}rfi, Michael Kohler, Adam Krzyzak, Harro Walk, et~al.
\newblock \emph{A distribution-free theory of nonparametric regression},
  volume~1.
\newblock Springer, 2002.

\bibitem[Hahn(1998)]{hahn1998role}
Jinyong Hahn.
\newblock On the role of the propensity score in efficient semiparametric
  estimation of average treatment effects.
\newblock \emph{Econometrica}, pages 315--331, 1998.

\bibitem[Hardt et~al.(2016)Hardt, Price, and Srebro]{hardt2016equality}
Moritz Hardt, Eric Price, and Nathan Srebro.
\newblock Equality of opportunity in supervised learning.
\newblock In \emph{Proceedings of the 30th International Conference on Neural
  Information Processing Systems}, pages 3323--3331, 2016.

\bibitem[Hirano and Porter(2012)]{hirano2012impossibility}
Keisuke Hirano and Jack~R Porter.
\newblock Impossibility results for nondifferentiable functionals.
\newblock \emph{Econometrica}, 80\penalty0 (4):\penalty0 1769--1790, 2012.

\bibitem[Huber(2004)]{huber2004robust}
Peter~J Huber.
\newblock \emph{Robust statistics}, volume 523.
\newblock John Wiley \& Sons, 2004.

\bibitem[Huber et~al.(1967)]{huber1967behavior}
Peter~J Huber et~al.
\newblock The behavior of maximum likelihood estimates under nonstandard
  conditions.
\newblock In \emph{Proceedings of the fifth Berkeley symposium on mathematical
  statistics and probability}, volume~1, pages 221--233. University of
  California Press, 1967.

\bibitem[Juraska et~al.(2012)Juraska, Gilbert, Lu, Zhang, Davidian, and
  Tsiatis]{juraska2010package}
M~Juraska, PB~Gilbert, X~Lu, M~Zhang, M~Davidian, and AA~Tsiatis.
\newblock speff2trial: Semiparametric efficient estimation for a two-sample
  treatment effect.
\newblock \emph{R package version}, 1\penalty0 (4), 2012.

\bibitem[Kallus et~al.(2019)Kallus, Mao, and Uehara]{kallus2019localized}
Nathan Kallus, Xiaojie Mao, and Masatoshi Uehara.
\newblock Localized debiased machine learning: Efficient estimation of quantile
  treatment effects, conditional value at risk, and beyond.
\newblock \emph{arXiv preprint arXiv:1912.12945}, 2019.

\bibitem[Kennedy(2016)]{kennedy2016semiparametric}
Edward~H Kennedy.
\newblock Semiparametric theory and empirical processes in causal inference.
\newblock In \emph{Statistical causal inferences and their applications in
  public health research}, pages 141--167. Springer, 2016.

\bibitem[Kennedy(2020)]{kennedy2020optimal}
Edward~H Kennedy.
\newblock Optimal doubly robust estimation of heterogeneous causal effects.
\newblock \emph{arXiv preprint arXiv:2004.14497}, 2020.

\bibitem[Kennedy et~al.(2020)Kennedy, Balakrishnan, G’Sell,
  et~al.]{kennedy2018sharp}
Edward~H Kennedy, Sivaraman Balakrishnan, Max G’Sell, et~al.
\newblock Sharp instruments for classifying compliers and generalizing causal
  effects.
\newblock \emph{Annals of Statistics}, 48\penalty0 (4):\penalty0 2008--2030,
  2020.

\bibitem[Kosorok and Laber(2019)]{kosorok2019precision}
Michael~R Kosorok and Eric~B Laber.
\newblock Precision medicine.
\newblock \emph{Annual review of statistics and its application}, 6:\penalty0
  263--286, 2019.

\bibitem[Laber and Murphy(2011)]{laber2011adaptive}
Eric~B Laber and Susan~A Murphy.
\newblock Adaptive confidence intervals for the test error in classification.
\newblock \emph{Journal of the American Statistical Association}, 106\penalty0
  (495):\penalty0 904--913, 2011.

\bibitem[Laber et~al.(2014)Laber, Lizotte, Qian, Pelham, and
  Murphy]{laber2014dynamic}
Eric~B Laber, Daniel~J Lizotte, Min Qian, William~E Pelham, and Susan~A Murphy.
\newblock Dynamic treatment regimes: Technical challenges and applications.
\newblock \emph{Electronic Journal of Statistics}, 8\penalty0 (1):\penalty0
  1225, 2014.

\bibitem[Linn et~al.(2017)Linn, Laber, and Stefanski]{linn2017interactive}
Kristin~A Linn, Eric~B Laber, and Leonard~A Stefanski.
\newblock Interactive q-learning for quantiles.
\newblock \emph{Journal of the American Statistical Association}, 112\penalty0
  (518):\penalty0 638--649, 2017.

\bibitem[Luedtke et~al.(2020)Luedtke, Chambaz, et~al.]{luedtke2020performance}
Alex Luedtke, Antoine Chambaz, et~al.
\newblock Performance guarantees for policy learning.
\newblock In \emph{Annales de l'Institut Henri Poincar{\'e}, Probabilit{\'e}s
  et Statistiques}, volume~56, pages 2162--2188. Institut Henri Poincar{\'e},
  2020.

\bibitem[Luedtke and van~der Laan(2016)]{luedtke2016statistical}
Alexander~R Luedtke and Mark~J van~der Laan.
\newblock Statistical inference for the mean outcome under a possibly
  non-unique optimal treatment strategy.
\newblock \emph{Annals of statistics}, 44\penalty0 (2):\penalty0 713, 2016.

\bibitem[Machado and Mata(2005)]{machado2005counterfactual}
Jos{\'e}~AF Machado and Jos{\'e} Mata.
\newblock Counterfactual decomposition of changes in wage distributions using
  quantile regression.
\newblock \emph{Journal of applied Econometrics}, 20\penalty0 (4):\penalty0
  445--465, 2005.

\bibitem[Mehrabi et~al.(2019)Mehrabi, Morstatter, Saxena, Lerman, and
  Galstyan]{mehrabi2019survey}
Ninareh Mehrabi, Fred Morstatter, Nripsuta Saxena, Kristina Lerman, and Aram
  Galstyan.
\newblock A survey on bias and fairness in machine learning.
\newblock \emph{arXiv preprint arXiv:1908.09635}, 2019.

\bibitem[Meinshausen(2006)]{meinshausen2006quantile}
Nicolai Meinshausen.
\newblock Quantile regression forests.
\newblock \emph{Journal of Machine Learning Research}, 7\penalty0
  (Jun):\penalty0 983--999, 2006.

\bibitem[Melly(2005)]{melly2005decomposition}
Blaise Melly.
\newblock Decomposition of differences in distribution using quantile
  regression.
\newblock \emph{Labour economics}, 12\penalty0 (4):\penalty0 577--590, 2005.

\bibitem[Mo et~al.(2020)Mo, Qi, and Liu]{mo2020learning}
Weibin Mo, Zhengling Qi, and Yufeng Liu.
\newblock Learning optimal distributionally robust individualized treatment
  rules.
\newblock \emph{Journal of the American Statistical Association}, pages 1--16,
  2020.

\bibitem[Murphy(2003)]{murphy2003optimal}
Susan~A Murphy.
\newblock Optimal dynamic treatment regimes.
\newblock \emph{Journal of the Royal Statistical Society: Series B (Statistical
  Methodology)}, 65\penalty0 (2):\penalty0 331--355, 2003.

\bibitem[Newey(1990)]{newey1990semiparametric}
Whitney~K Newey.
\newblock Semiparametric efficiency bounds.
\newblock \emph{Journal of Applied Econometrics}, 5\penalty0 (2):\penalty0
  99--135, 1990.

\bibitem[Nie and Wager(2021)]{nie2021quasi}
Xinkun Nie and Stefan Wager.
\newblock Quasi-oracle estimation of heterogeneous treatment effects.
\newblock \emph{Biometrika}, 108\penalty0 (2):\penalty0 299--319, 2021.

\bibitem[Qi et~al.(2019)Qi, Pang, and Liu]{qi2019estimating}
Zhengling Qi, Jong-Shi Pang, and Yufeng Liu.
\newblock Estimating individualized decision rules with tail controls.
\newblock \emph{arXiv preprint arXiv:1903.04367}, 2019.

\bibitem[Quionero-Candela et~al.(2009)Quionero-Candela, Sugiyama, Schwaighofer,
  and Lawrence]{quionero2009dataset}
Joaquin Quionero-Candela, Masashi Sugiyama, Anton Schwaighofer, and Neil~D
  Lawrence.
\newblock \emph{Dataset shift in machine learning}.
\newblock The MIT Press, 2009.

\bibitem[Robins et~al.(2008)Robins, Li, Tchetgen, van~der Vaart,
  et~al.]{robins2008higher}
James Robins, Lingling Li, Eric Tchetgen, Aad van~der Vaart, et~al.
\newblock Higher order influence functions and minimax estimation of nonlinear
  functionals.
\newblock In \emph{Probability and statistics: essays in honor of David A.
  Freedman}, pages 335--421. Institute of Mathematical Statistics, 2008.

\bibitem[Robins(2004)]{robins2004optimal}
James~M Robins.
\newblock Optimal structural nested models for optimal sequential decisions.
\newblock In \emph{Proceedings of the second seattle Symposium in
  Biostatistics}, pages 189--326. Springer, 2004.

\bibitem[Rockafellar and Uryasev(2002)]{rockafellar2002conditional}
R~Tyrrell Rockafellar and Stanislav Uryasev.
\newblock Conditional value-at-risk for general loss distributions.
\newblock \emph{Journal of Banking \& Finance}, 26\penalty0 (7):\penalty0
  1443--1471, 2002.

\bibitem[Rothe(2010)]{rothe2010nonparametric}
Christoph Rothe.
\newblock Nonparametric estimation of distributional policy effects.
\newblock \emph{Journal of Econometrics}, 155\penalty0 (1):\penalty0 56--70,
  2010.

\bibitem[Rubin(1974)]{rubin1974estimating}
Donald~B Rubin.
\newblock Estimating causal effects of treatments in randomized and
  nonrandomized studies.
\newblock \emph{Journal of Educational Psychology}, 66\penalty0 (5):\penalty0
  688, 1974.

\bibitem[Schulte et~al.(2014)Schulte, Tsiatis, Laber, and
  Davidian]{schulte2014q}
Phillip~J Schulte, Anastasios~A Tsiatis, Eric~B Laber, and Marie Davidian.
\newblock Q-and a-learning methods for estimating optimal dynamic treatment
  regimes.
\newblock \emph{Statistical Science: a review Journal of the Institute of
  Mathematical Statistics}, 29\penalty0 (4):\penalty0 640, 2014.

\bibitem[Silverman(1986)]{silverman1986density}
Bernard~W Silverman.
\newblock \emph{Density estimation for statistics and data analysis},
  volume~26.
\newblock CRC press, 1986.

\bibitem[Stone(1977)]{stone1977consistent}
Charles~J Stone.
\newblock Consistent nonparametric regression.
\newblock \emph{The annals of statistics}, pages 595--620, 1977.

\bibitem[Tsiatis(2007)]{tsiatis2007semiparametric}
Anastasios Tsiatis.
\newblock \emph{Semiparametric theory and missing data}.
\newblock Springer Science \& Business Media, 2007.

\bibitem[Tsybakov et~al.(2004)]{tsybakov2004optimal}
Alexander~B Tsybakov et~al.
\newblock Optimal aggregation of classifiers in statistical learning.
\newblock \emph{The Annals of Statistics}, 32\penalty0 (1):\penalty0 135--166,
  2004.

\bibitem[van~der Laan and Luedtke(2014)]{van2014targeted}
Mark~J van~der Laan and Alexander~R Luedtke.
\newblock Targeted learning of an optimal dynamic treatment, and statistical
  inference for its mean outcome.
\newblock 2014.

\bibitem[{van der Laan} and Robins(2003)]{van2003unified}
Mark~J {van der Laan} and James~M Robins.
\newblock \emph{Unified methods for censored longitudinal data and causality}.
\newblock Springer Science \& Business Media, 2003.

\bibitem[van~der Vaart(2000)]{van2000asymptotic}
Aad~W van~der Vaart.
\newblock \emph{Asymptotic statistics}, volume~3.
\newblock Cambridge university press, 2000.

\bibitem[van~der Vaart(2002)]{van2002semiparametric}
Aad~W van~der Vaart.
\newblock Semiparametric statistics.
\newblock \emph{Lecture Notes in Math.}, \penalty0 (1781), 2002.

\bibitem[Wang et~al.(2018)Wang, Zhou, Song, and Sherwood]{wang2018quantile}
Lan Wang, Yu~Zhou, Rui Song, and Ben Sherwood.
\newblock Quantile-optimal treatment regimes.
\newblock \emph{Journal of the American Statistical Association}, 113\penalty0
  (523):\penalty0 1243--1254, 2018.

\bibitem[Wasserman(2006)]{wasserman2006all}
Larry Wasserman.
\newblock \emph{All of nonparametric statistics}.
\newblock Springer Science \& Business Media, 2006.

\bibitem[Xiao et~al.(2019)Xiao, Zhang, and Lu]{xiao2019robust}
Wei Xiao, Hao~Helen Zhang, and Wenbin Lu.
\newblock Robust regression for optimal individualized treatment rules.
\newblock \emph{Statistics in medicine}, 38\penalty0 (11):\penalty0 2059--2073,
  2019.

\bibitem[Zhang et~al.(2012)Zhang, Tsiatis, Laber, and
  Davidian]{zhang2012robust}
Baqun Zhang, Anastasios~A Tsiatis, Eric~B Laber, and Marie Davidian.
\newblock A robust method for estimating optimal treatment regimes.
\newblock \emph{Biometrics}, 68\penalty0 (4):\penalty0 1010--1018, 2012.

\bibitem[Zhang et~al.(2013)Zhang, Tsiatis, Laber, and
  Davidian]{zhang2013robust}
Baqun Zhang, Anastasios~A Tsiatis, Eric~B Laber, and Marie Davidian.
\newblock Robust estimation of optimal dynamic treatment regimes for sequential
  treatment decisions.
\newblock \emph{Biometrika}, 100\penalty0 (3):\penalty0 681--694, 2013.

\bibitem[Zhang et~al.(2021)Zhang, Troxel, and Petkova]{zhang2021robust}
Jinchun Zhang, Andrea~B Troxel, and Eva Petkova.
\newblock Robust index of confidence weighted learning for optimal
  individualized treatment rule estimation.
\newblock \emph{Stat}, 10\penalty0 (1):\penalty0 e374, 2021.

\bibitem[Zhao et~al.(2012)Zhao, Zeng, Rush, and Kosorok]{zhao2012estimating}
Yingqi Zhao, Donglin Zeng, A~John Rush, and Michael~R Kosorok.
\newblock Estimating individualized treatment rules using outcome weighted
  learning.
\newblock \emph{Journal of the American Statistical Association}, 107\penalty0
  (499):\penalty0 1106--1118, 2012.

\bibitem[Zheng and van~der Laan(2010)]{zheng2010asymptotic}
Wenjing Zheng and Mark~J van~der Laan.
\newblock Asymptotic theory for cross-validated targeted maximum likelihood
  estimation.
\newblock 2010.

\end{thebibliography}

\clearpage
\appendix
\section{Proofs}
\label{sec:proofs}

\subsection{Proofs in Section~\ref{sec:policy_comparision}}

\begin{proof}[Proof of Proposition~\ref{prop:ate-acme-variance}] 
Pick $x \in \mathcal{X}$. 
For all $a \in \{0, 1\}$, 
we have $|\mu_a(x) - m_a(x)| \leq \sigma_{a}(x)$. 
If $\mu_1(x) > \mu_0(x)$, 
then $m_1(x) \geq \mu_1(x) - \sigma_1(x) 
>  \mu_0(x) + \sigma_0(x) \geq m_0(x)$, 
where the second inequality holds because of the assumption 
that $|\mu_1(x) - \mu_0(x)| > \sigma_1 (x) + \sigma_0 (x)$.
Similarly, 
$m_1(x) > m_0(x)$  
implies that $\mu_1(x) \geq m_1(x) - \sigma_1(x) 
>  m_0(x) + \sigma_0(x) \geq \mu_0(x)$.
Thus, we have that $\mu_1\left(x\right) > \mu_0\left(x \right) \Leftrightarrow
m_1 \left( x \right) > m_0 \left(x \right)$, which completes the proof.
\end{proof}

\begin{proof}[Proof of Proposition~\ref{prop:GoM}]
By definition, the marginal median optimal policy $d_\text{MME}^*$ satisfies that for all $d \in \mathcal{D}$, $m(Y^{d_\text{MME}^*}) \geq m(Y^d)$. 
We use a tuple $(a_1, a_0)$ to represent a policy $d$ where $a_i = d(i)$,
which gives us the following table:
\begin{table}[H]
\centering
\begin{tabular}{@{}cc@{}}
\toprule
$(a_0, a_1)$     & $m(Y^d)$ \\ \midrule
$(0,1)$ &    $\frac{\sigma_{1}(1)\mu_{0}(0) + \sigma_{0}(0)\mu_{1}(1)}{\sigma_{0}(0) + \sigma_{1}(1)}$      \\ \midrule
$(1,1)$ &       $\frac{\sigma_{1}(1)\mu_{1}(0) + \sigma_{1}(0)\mu_{1}(1)}{\sigma_{1}(0) + \sigma_{1}(1)}$   \\ \midrule
$(0,0)$ &       $\frac{\sigma_{0}(1)\mu_{0}(0) + \sigma_{0}(0)\mu_{0}(1)}{\sigma_{0}(0) + \sigma_{0}(1)}$   \\ \midrule
$(1,0)$ &      $\frac{\sigma_{0}(1)\mu_{1}(0) + \sigma_{1}(0)\mu_{0}(1)}{\sigma_{1}(0) + \sigma_{0}(1)}$    \\ \bottomrule
\end{tabular}
\caption{Marginal medians of all possible policies in the setting described in Section~\ref{sec:d_acme_d_mte}.}
\end{table}

Given $\sigma_1(1) > \sigma_0(1)$ and 
$\mu_1(1) > \mu_0(1)$ and $\frac{\mu_1(1)}{\mu_0(1)} < \frac{\sigma_1(1)}{\sigma_0(1)}$,
when 
we pick $\{\mu_a(0), \sigma_a(0)\}_{a=0}^1$ such that 
for $a =0$ and $a=1$, 
\begin{align}
    &\left(\sigma_1(1) - \sigma_0(1) \right)\mu_a(0)
    + \left(\mu_1(1) - \mu_0(1) \right) \sigma_a(0)
    > \sigma_1(1) \mu_0(1) - \sigma_0(1) \mu_1(1),
\label{eq:mte_opt_criteria}
\end{align}
the optimal policy $d^*_\text{MME}$ will assign the decision for $x=1$ according to  
$d^*_\text{MME}(1) = \mathbb{1}\{\mu_1(1) > \mu_0(1)\}$
since the conditions have ensured that 
$m(Y^{(1,1)}) > m(Y^{(1,0)})$ 
and $m(Y^{(0,1)}) > m(Y^{(0,0)})$. 
On the other hand, if~\eqref{eq:mte_opt_criteria} does not hold
for $a =0$ and $a=1$, 
then the marginal median optimal policy 
$d^*_\text{MME}(1) = \mathbb{1}\{\mu_1(1) \leq \mu_0(1)\}$. 
\end{proof}

\subsection{Proofs in Section~\ref{sec:efficiency}}
\label{appendix:efficiency-proofs}
\paragraph{Derivation of the influence function under discrete covariates and continuous outcome.}
Let $\text{IF}(\cdot)$ denote the operator that returns the  influence function of an input functional. 
In this analysis, we assume $X$ to be discrete and $Y$ to be continuous. 
Let $p(x)$ denote the probability mass function of $X$. 
We start with the chain rule of $\text{IF}(\cdot)$.
\begin{align*}
    &\phi_d(Z) = \text{IF}(\psi_d(\mathbb{P}))  
    = \text{IF}\left(\sum_{x \in \mathcal{X}} p(x) m_d(x) \right)\\
    &= \sum_{x \in \mathcal{X}} p(x) \Big\{\text{IF}(m_1(x)) d(x)  + \text{IF}(m_0(x))(1-d(x)) \Big\} 
    +\text{IF}(p(x)) m_d(x),
\end{align*}

where 
$\text{IF}(p(x)) = \mathbb{1}\{X=x\} - p(x)$. 
Thus, it suffices to find $\text{IF}(m_a(x))$ where $a \in \{0,1\}$. 
 
Let $\delta_z$ denote the Dirac measure at $z$.
We use the submodel $\mathbb{P}_\epsilon(z) = (1-\epsilon)\mathbb{P}(z) + \delta_{z'}$ for some $z' = (x', a', y')$ to find the influence function, which gives that 
\begin{align*}
    f_\epsilon(y|a,x) = \frac{(1-\epsilon) f(y|x,a)\mathbb{P}(X=x, A=a) + \epsilon \delta_{z'}}{(1-\epsilon) \mathbb{P}(X=x, A=a) + \epsilon \mathbb{1}\{x=x', a=a'\}}, 
\end{align*}
where $f_\epsilon$ and $f$ are densities with respect to 
$\mathbb{P}_\epsilon$ and $\mathbb{P}$ respectively. 
Let $m_{a,\epsilon}(x)$ denote the median of the distribution $\mathbb{P}_\epsilon(Y|X=x, A=a)$.
By the definition of the median, we have that
\begin{align*}
     F_\epsilon(m_{a,\epsilon}(x) |x, a) &= \int_{y \leq  m_{a,\epsilon}(x)} f_\epsilon (y|x, a) dy \\
     &=\int_{y \leq  m_{a,\epsilon}(x)}  \frac{(1-\epsilon) f(y|x,a)\mathbb{P}(X=x,A=a) + \epsilon \delta_{z'}}{(1-\epsilon) \mathbb{P}(X=x, A=a) + \epsilon \mathbb{1}\{x=x', a=a'\}} dy = \frac{1}{2} 
\end{align*}

When $x \neq x'$ or $a \neq a'$, we have that 
\begin{align*}
    &\mathbb{P}(X=x, A=a) \int_{y \leq  m_{a,\epsilon}(x)}  f(y| x,a)   dy = \frac{\mathbb{P}(X=x, A=a)}{2}, 
\end{align*}
which suggests that $m_a(x) = m_{a, \epsilon}(x)$. 
When $x = x'$ and $a= a'$, we have that 
\begin{align}
    \int_{y \leq  m_{a,\epsilon}(x)}  \frac{(1-\epsilon) f(y|x,a)\mathbb{P}(X=x, A=a) + \epsilon \delta_{y'}}{(1-\epsilon) \mathbb{P}(X=x, A=a) + \epsilon} dy = \frac{1}{2}.
\label{eq:median-influence-fn-condition}
\end{align}
It has two cases 
\begin{itemize}
    \item When $y' > m_{a,\epsilon}(x)$, 
    \eqref{eq:median-influence-fn-condition} can be simplified into 
    \begin{align*}
       F(m_{a,\epsilon}(x)|x,a) = \int_{y \leq  m_{a,\epsilon}(x)}   f(y|x,a)dy = \frac{(1-\epsilon) \mathbb{P}(X=x, A=a) + \epsilon}{2(1-\epsilon)\mathbb{P}(X=x, A=a) },
    \end{align*}
    and $m_{a,\epsilon}(x) = F^{-1}\left( \frac{(1-\epsilon) \mathbb{P}(X=x, A=a) + \epsilon}{2(1-\epsilon)\mathbb{P}(X=x, A=a) } \bigg| x,a \right)$.
    \item When $y' \leq m_{a,\epsilon}(x)$, 
    ~\eqref{eq:median-influence-fn-condition} can be simplified into
    \begin{align*}
       F(m_{a,\epsilon}(x)|x,a) = \int_{y \leq  m_{a,\epsilon}(x)}   f(y|x,a)dy = \frac{(1-\epsilon) \mathbb{P}(X=x, A=a) - \epsilon}{2(1-\epsilon)\mathbb{P}(X=x, A=a) }, 
    \end{align*}
    and $m_{a,\epsilon}(x) = F^{-1}\left( \frac{(1-\epsilon) \mathbb{P}(X=x, A=a) - \epsilon}{2(1-\epsilon)\mathbb{P}(X=x, A=a) } \bigg| x,a \right)$.
\end{itemize}

Putting it altogether, we have that when $x \neq x'$ or $a \neq a'$:
\begin{align*}
    \text{IF}(m_{a}(x)) = \frac{d}{d \epsilon} m_{a,\epsilon}(x) \Big|_{\epsilon=0} = 0 =  \mathbb{1}\{x=x'\} \mathbb{1}\{a=a'\}.
\end{align*}

When  $x = x'$ and $a = a'$, we have that
\begin{align*}
   \text{IF}(m_{a}(x)) = \mathbb{1}\{x=x'\} \mathbb{1}\{a=a'\}\frac{d}{d \epsilon} m_{a,\epsilon}(x) \Big|_{\epsilon=0}
   =-\frac{\mathbb{1}\{x=x'\} \mathbb{1}\{a=a'\}}{\mathbb{P}(X=x)\mathbb{P}(A=a|X=x)}
   \frac{\mathbb{1}\{y' \leq m_{a}(x)\} - 1/2}{ f(m_{a}(x)|x,a)},
\end{align*}
where the last equality holds since for all $m_{a,\epsilon}(x) = F^{-1}(g(\epsilon)|x,a)$, 
\begin{align*}
    \frac{d}{d \epsilon} m_{a,\epsilon}(x) \Big|_{\epsilon=0}
    &= \left(\frac{1}{f(m_{a,\epsilon}(x) |x,a)}\frac{d}{d\epsilon}g(\epsilon) \right) \bigg|_{\epsilon=0},
\end{align*}
and for example when $g(\epsilon) = \frac{(1-\epsilon)\mathbb{P}(X=x,A=a)+\epsilon}{2(1-\epsilon)\mathbb{P}(X=x, A=a)}$, i.e., when $y' > F^{-1}\Big(\frac{(1-\epsilon)\mathbb{P}(X=x,A=a)+\epsilon}{2(1-\epsilon)\mathbb{P}(X=x, A=a)} \bigg| x, a \Big)$, 
gives that 
$
    \frac{d}{d \epsilon} g(\epsilon)
    = \frac{1}{2 \mathbb{P}(X=x, A=a)} \cdot \frac{d}{d \epsilon} \frac{\epsilon}{1-\epsilon} =  \frac{1}{2(1-\epsilon)^2 \mathbb{P}(X=x, A=a)}. 
$

Therefore, we have obtained that
\begin{align*}
    \phi_d(Z) &=\sum_{x \in \mathcal{X}} p(x) \Big\{ \text{IF}\left(m_1(x) \right)d(x) + \text{IF} \left(m_0(x) \right)(1-d(x)) \Big\}
    + \text{IF}(p(x)) m_{d}(X)\\
    &= \sum_{x \in \mathcal{X}} p(x) 
    \Big\{ \frac{\mathbb{1}\{X=x\} \mathbb{1}\{A=1\}}{\mathbb{P}(X=x)\mathbb{P}(A=1|X=x)}
   \frac{1/2 - \mathbb{1}\{Y \leq m_{1}(x)\}}{ f(m_{1}(x) |x,1)}d(x) \\
  & +\frac{\mathbb{1}\{X=x\} \mathbb{1}\{A=0\}}{\mathbb{P}(X=x)\mathbb{P}(A=0|X=x)}
   \frac{1/2 - \mathbb{1}\{Y \leq m_{0}(x)\}}{ f(m_{0}(x)|x,0)}
   (1-d(x)) \Big\}
    + (\mathbb{1}\{X=x\} - p(x)) m_{d}(x)\\   
    &=  \frac{Ad(X)}{\pi_1(X)}\frac{1/2 - \mathbb{1}\{Y \leq m_1(X) \}}{ f_{1,m}(X)} + \frac{(1-A)(1-d(X))}{\pi_0(X)}\frac{1/2 - \mathbb{1}\{Y \leq m_0(X) \}}{ f_{0,m}(X)} + m_{d}(X)- \psi_d.
\end{align*}
where the last equality follows since the indicator function $\mathbb{1}\{X=x\}$ picks out the values $x \in \mathcal{X}$ such that it equals $X$ and $p(x) = \mathbb{P}(X=x)$.

\begin{proof}[Proof of Lemma~\ref{lemma:decomposition}]
By definition, since $\int \phi_d(Z;\overline{\mathbb{P}}) d \overline{\mathbb{P}} = 0$, we have that 
\begin{align*}
    R_d(\overline{\mathbb{P}}, \mathbb{P})
    &= \psi_d(\overline{\mathbb{P}}) - \psi_d(\mathbb{P}) + \int \phi_d(Z;\overline{\mathbb{P}}) d \mathbb{P}\\
    &= - \psi_d(\mathbb{P}) 
    + \int_\mathcal{X} \frac{\pi_1(x) d(x)}{\overline{\pi}_1(x)}\frac{1/2 -  F_{1,\overline{m}_1}(x)}
    {\overline{f}_{1,\overline{m}}(x)} 
    + \frac{\pi_0(x)(1-d(x))}{\overline{\pi}_0(x)}\frac{1/2 - F_{0, \overline{m}}(x) }{\overline{f}_{0,\overline{m}}(x)} 
    + \overline{m}_{d}(x) d \mathbb{P}(x)\\
   &= \int_{\mathcal{X}} 
    d(x)\Big({\overline{m}_1(x) - m_1(x) - \frac{\pi_1(x)}{\overline{\pi}_1(x)}
    {\frac{F_{1,\overline{m}}(x) - {F}_{1,m}(x)  }{\overline{f}_{1,\overline{m}}(x)}}}
    \Big)\\
    &\qquad + (1 - d(x)) \Big({\overline{m}_0(x) -m_0(x) - \frac{\pi_0(x)}{\overline{\pi}_0(x)}
    \frac{F_{0, \overline{m}}(x) 
    - {F}_{0,m}(x)}{\overline{f}_{0,\overline{m}}(x)}}
    \Big)d \mathbb{P}(x), 
\end{align*}
where the second equality follows from 
iterated expectation, i.e., $\mathbb{E}_{\mathbb{P}}[\mathbb{1}\{Y \leq \overline{m}_a(X)\}|X=x,A=a] = F_{a,\overline{m}}(x)$, 
and 
the third equality follows from rearranging terms 
and the fact that $F_{a,m}(x) = 1/2$. 
\end{proof}

\begin{proof}[Proof of Corollary~\ref{corollary:R_d_product}]
Recall that for~\eqref{eq:R_d_raw}, we have   
\begin{align*}
    R_d(\overline{\mathbb{P}}, \mathbb{P})
    &= \int_{\mathcal{X}} 
    d(x)\Big(\underbrace{\overline{m}_1(x) - m_1(x) - \frac{\pi_1(x)}{\overline{\pi}_1(x)}
    \overbrace{\frac{F_{1,\overline{m}}(x) - {F}_{1,m}(x)  }{\overline{f}_{1,\overline{m}}(x)}}^{A_2}}_{A_1}
    \Big)\\
    &\qquad + (1 - d(x)) \Big(\underbrace{\overline{m}_0(x) -m_0(x) - \frac{\pi_0(x)}{\overline{\pi}_0(x)}
    \frac{F_{0, \overline{m}}(x) 
    - {F}_{0,m}(x)}{\overline{f}_{0,\overline{m}}(x)}}_{A_0}
    \Big)d \mathbb{P}(x), 
\end{align*}
For $x \in \mathcal{X}$ 
such that $f_a(y|x)$ is differentiable and $L$-Lipschitz continuous in  $y$, 
we can simplify $A_2$ as:
\begin{align*}
    A_2 
    &=\frac{F_{1,\overline{m}}(x) 
    - {F}_{1,{m}}(x)}
    {\overline{f}_{1,\overline{m}}(x)}\\ 
    &= \frac{1}{\overline{f}_{1,\overline{m}}(x)}
    \left( F_{1,m}(x)
    + f_{1,m}(x) 
    (\overline{m}_1(x) - m_1(x))
    + \frac{f_1'(c_{1}(x)|x)}{2} (\overline{m}_1(x) - m_1(x))^2 
    - F_{1,m}(x)\right)\\
    &=\frac{1}{\overline{f}_{1,\overline{m}}(x)}
    \left(f_{1,m}(x) 
    (\overline{m}_1(x) - m_1(x)) + 
     \frac{f'_{1,c}(x)}{2} (\overline{m}_1(x) - m_1(x))^2 \right), 
\end{align*}
where the second equality follows from Taylor's theorem
and $c_{1}(x)$ is a value in between $\overline{m}_1(x)$ and ${m}_1(x)$. 
To simplify the notations, in the third equality (and the following usage of Taylor's theorem), 
we denote the derivative of the conditional probability
density function $f_a(\cdot| X=x)$ at $c_{a}(x)$
to be $f'_{a,c}(x)$ 
where $c_{a}(x)$ is a value between $\overline{m}_a(x)$ and ${m}_a(x)$
and satisfies that 
$F_{a,\overline{m}}(x) = F_{a,m}(x) + f_{a,m}(x)(\overline{m}_a(x) - m_a(x))
+ f'_{a,c}(x)(\overline{m}_a(x) - m_a(x))^2/2$. 
This gives that 
\begin{align*}
    A_1 &= 
    \overline{m}_1(x) -m_1(x) 
    - \frac{\pi_1(x)}{\overline{\pi}_1(x)} 
   \frac{1}{\overline{f}_{1,\overline{m}}(x)}
    \left(f_{1,m}(x) 
    (\overline{m}_1(x) - m_1(x)) + 
     \frac{f'_{1,c}(x)}{2} (\overline{m}_1(x) - m_1(x))^2 \right)\\
    &=  \left(\overline{m}_1(x) -m_1(x) \right)
    \left(1 - \frac{\pi_1(x)}{\overline{\pi}_1(x)}  \frac{f_{1,m}(x)}
    {\overline{f}_{1, \overline{m}}(x)}\right) 
    - \frac{f'_{1,c}(x) \pi_1(x) (\overline{m}_1(x) - m_1(x))^2}{2\overline{\pi}_1(x)\overline{f}_{1,\overline{m}}(x)}.
\end{align*}

Applying the same logic to $A_0$, we obtain that 
\begin{align*}
    &R_d(\overline{\mathbb{P}}, \mathbb{P})
    = \int d(x) \left( \left(\overline{m}_1(x) -m_1(x) \right)
    \left(1 - \frac{\pi_1(x)}{\overline{\pi}_1(x)}  \frac{f_{1,m}(x)}
    {\overline{f}_{1, \overline{m}}(x)}\right) 
    - \frac{f'_{1,c}(x) \pi_1(x) (\overline{m}_1(x) - m_1(x))^2}{2\overline{\pi}_1(x)\overline{f}_{1,\overline{m}}(x)} \right)\\
    &+ (1-d(x)) \left(
    \left(\overline{m}_0(x) -m_0(x) \right)
    \left(1 - \frac{\pi_0(x)}{\overline{\pi}_0(x)}  \frac{f_{0,m}(x)}
    {\overline{f}_{0, \overline{m}}(x)}\right) 
    - \frac{f'_{0,c}(x) \pi_0(x) (\overline{m}_0(x) - m_0(x))^2}{2 \overline{\pi}_0(x) \overline{f}_{0,\overline{m}}(x)} \right)
    d \mathbb{P}(x)\\
    &= \mathbb{P}\Bigg\{ 
    \frac{d}{\overline{\pi}_1  \overline{f}_{1,\overline{m}}} 
    \left(
    \left(\overline{m}_1 -m_1 \right)
    \left(\overline{\pi}_1  \overline{f}_{1,\overline{m}} - 
    \pi_1 f_{1,m}
    \right) 
    - \frac{f'_{1,c} \pi_1 (\overline{m}_1 - m_1)^2}{2}\right) \\
    &\qquad + \frac{(1 - d)}{ \overline{\pi}_0  \overline{f}_{0,\overline{m}}} 
    \left(
    \left(\overline{m}_0 -m_0 \right)
    \left( \overline{\pi}_0  \overline{f}_{0,\overline{m}} - \pi_0 f_{0,m}\right) 
    -  \frac{f'_{0,c} \pi_0 (\overline{m}_0 - m_0)^2}{2 }\right)
    \Bigg\}\\
    &= \mathbb{P}\Bigg\{ 
    \frac{d\left(\overline{m}_1 -m_1 \right)}{\overline{\pi}_1  \overline{f}_{1,\overline{m}}} 
    \left(
    \left(\overline{\pi}_1 - \pi_1 \right)  \overline{f}_{1,\overline{m}} 
    +  
    \left(  \overline{f}_{1,\overline{m}} - f_{1,m}\right) \pi_1
    - (\overline{m}_1 - m_1)\frac{f'_{1,c} \pi_1 }{2}\right) \\
    &\qquad + \frac{(1 - d)\left(\overline{m}_0 -m_0 \right)}{\overline{\pi}_0  \overline{f}_{0,\overline{m}}} 
    \left(
    (\overline{\pi}_0 - \pi_0)  \overline{f}_{0,\overline{m}} 
    +(\overline{f}_{0,\overline{m}} - f_{0,m})\pi_0 
    - (\overline{m}_0 - m_0) \frac{f'_{0,c} \pi_0 }{2 }\right) \Bigg\}.
\end{align*}
\end{proof}

\begin{proof}[Proof of Corollary~\ref{corollary:influence-function-calculation}]
Consider a parametric submodel $\mathbb{P}_\epsilon$, 
e.g., the probability density functions
$p_\epsilon$ and $p$ of $\mathbb{P}_\epsilon$ and $\mathbb{P}$ 
satisfy that 
$p_\epsilon(z) = p(z) (1 +\epsilon h(z))$ where $\mathbb{E}[h] = 0$, 
$\|h\|_\infty < M$ and $\epsilon > 1/M$.  
To show that $\phi_d$ is an influence function, 
it suffices to show that the mean-zero $\phi_d$ satisfies path-wise differentiability, 
i.e., 
\begin{align*}
    \frac{d \psi_d(\mathbb{P}_\epsilon)}{d \epsilon} \Bigg|_{\epsilon=0}
    = \int_Z \phi_d(Z;\mathbb{P}) \left( \frac{d}{d \epsilon} \log d\mathbb{P}_\epsilon \right)\Bigg|_{\epsilon = 0} d\mathbb{P}.
\end{align*}
As shown in Lemma~\ref{lemma:decomposition}, we have that 
$ \psi_d(\mathbb{P}_\epsilon) = \psi_d(\mathbb{P}) + \int_Z \phi_d(Z;\mathbb{P}) d (\mathbb{P}_\epsilon - \mathbb{P}) - R_d(\mathbb{P}, \mathbb{P}_\epsilon)$, 
which gives that 
\begin{align*}
    \frac{d \psi_d(\mathbb{P}_\epsilon)}{d \epsilon} \Bigg|_{\epsilon=0}
    &=  \frac{d}{d \epsilon}\int \phi_d(Z;\mathbb{P}) d \mathbb{P}_\epsilon \Bigg|_{\epsilon = 0}
    - \; \frac{d}{d \epsilon}R_d(\mathbb{P}, \mathbb{P}_\epsilon) \Bigg|_{\epsilon = 0}\\
    &=\int_Z \phi_d(Z;\mathbb{P}) \left( \frac{d}{d \epsilon} \log d\mathbb{P}_\epsilon \right)\Bigg|_{\epsilon = 0} d\mathbb{P}_\epsilon 
    - \; \frac{d}{d \epsilon}R_d(\mathbb{P}, \mathbb{P}_\epsilon) \Bigg|_{\epsilon = 0}.
\end{align*}
Finally, as shown in Corollary~\ref{corollary:R_d_product}, 
$R_d(\mathbb{P}, \mathbb{P}_\epsilon)$ only contains second-order products of errors between $\mathbb{P}$ and $\mathbb{P}_\epsilon$ and thus using the product rule on $\frac{d}{d \epsilon}R_d(\mathbb{P}, \mathbb{P}_\epsilon)$, one can get that  
\begin{align*}
    \frac{d}{d \epsilon}R_d(\mathbb{P}, \mathbb{P}_\epsilon) \Bigg|_{\epsilon = 0} = 0.
\end{align*}
Since our model is nonparametric, there is only one  influence function which is efficient~\citep{bickel1993efficient,tsiatis2007semiparametric}. Thus, $\phi_d$ is the efficient influence function of $\psi_d$.
\end{proof}

\begin{proof}[Proof of Theorem~\ref{thm:variance}] 
For a given policy $d \in \mathcal{D}$, we have that 
\begin{align*}
    &\text{Var}\left\{ \phi_d (Z)\right\} 
    = \mathbb{E}[\phi_d(Z)^2] = \mathbb{E}[(\xi_d(Z) - \psi_d)^2] \\
    =& \mathbb{E}\left\{
    \left(\frac{Ad(X)}{\pi_1(X)}\frac{1/2 - \mathbb{1}\{Y \leq m_1(X) \}}{ f_{1,m}\left(X\right)} \right)^2
    + \left(\frac{(1-A)(1-d(X))}{\pi_0(X)}\frac{1/2 - \mathbb{1}\{Y \leq m_0(X) \}}{ 
    f_{0,m}\left(X \right)}\right)^2 + 
    \left(m_{d}(X)- \psi_d\right)^2
    \right\} \\
    =& \mathbb{E}\left\{
    \left(\frac{Ad(X)}{\pi_1(X)}\frac{1/2 - \mathbb{1}\{Y \leq m_1(X) \}}{ f_{1,m}\left(X\right)} \right)^2 \right\}
    + \mathbb{E}\left\{\left(\frac{(1-A)(1-d(X))}{\pi_0(X)}\frac{1/2 - \mathbb{1}\{Y \leq m_0(X) \}}{ 
    f_{0,m}\left(X\mid X, 0\right)}\right)^2 \right\}
    + \text{Var}\{m_d(X)\}%
    \\
    =& \mathbb{E}_{X,A}\left\{
    \frac{Ad(X)}{\pi_1^2(X)}\frac{\mathbb{E}_{Y|X,A}\left[\left(1/2 - \mathbb{1}\{Y \leq m_1(X) \}\right)^2\right]}
    { f_{1,m}^2\left(X\right)}  \right\}\\
    &\qquad + \mathbb{E}_{X,A}\left\{
    \frac{(1-A)(1-d(X))}{\pi_0(X)^2}
    \frac{\mathbb{E}_{Y|X,A}\left[\left(1/2 - \mathbb{1}\{Y \leq m_0(X) \}\right)^2 \right]}{ 
    f_{0,m}\left(X\right)}\right\}
    + \text{Var}\{m_d(X)\}\\
    =& \mathbb{E}_X\left\{ 
    \frac{d(X)\text{Var}\{\mathbb{1}\{Y \leq m_1(X) | X, A=1\}\}}{\pi_1(X)f_{1,m}^2(X)} 
    + \frac{(1 - d(X))\text{Var}\{\mathbb{1}\{Y \leq m_0(X) \}|X, A=0\}}{\pi_0(X)f_{0,m}^2(X)} 
    \right\} 
     + \text{Var}\{m_d(X)\}\\
    =&  \mathbb{E}\left\{ 
    \frac{d(X)}{4\pi_1(X)f_{1,m}^2(X)} 
    + \frac{(1 - d(X))}{4 \pi_0(X) f_{0,m}^2(X)} \right\} 
    + \text{Var}\left\{ m_d(X) \right\},
\end{align*}
where the first equality is true since $\mathbb{E}[\phi_d(Z)] = 0$, 
the forth equality uses the fact $A^2 = A$, 
the fifth equality holds as $\mathbb{E}_{Y|X,A=a}[\mathbb{1}\{Y\leq m_a(X)\}] = 1/2$, 
and the last equality is true since $\text{Var}\{\mathbb{1}\{Y \leq m_1(X) |X, A=1\} = \mathbb{P}(Y \leq m_1(X) |X, A=1) \left(1 - \mathbb{P} (Y \leq m_1(X)|X, A=1) \right) = 1/4$. 
\end{proof}

\subsection{Proofs in Section~\ref{sec:estimation}}
\label{appendix:estimator_proof}

\begin{proof}[Proof of Theorem~\ref{thm:fixed_policy_estimator}] 
We start with the following decomposition  
\begin{align*}
    \widehat{\psi}_{d, \text{dr}} -  \psi_d
    &= \mathbb{P}_n \xi_d(\widehat{\mathbb{P}}) - \mathbb{P} \xi_d(\widehat{\mathbb{P}})  
    = \underbrace{\big(\mathbb{P}_n - \mathbb{P} \big) \left(\xi_d(\widehat{\mathbb{P}}) - \xi_d(\mathbb{P})\right)}_{T_1} 
    + \underbrace{\big(\mathbb{P}_n - \mathbb{P} \big) \xi_d(\mathbb{P})}_{T_2} 
    + \underbrace{\mathbb{P}\left(\xi_d(\widehat{\mathbb{P}}) - \xi_d(\mathbb{P}) \right)}_{T_3}. 
\end{align*}

When $\widehat{\xi}_d$ is learned from a separate sample from the empirical measure $\mathbb{P}_n$, 
by Lemma~\ref{lemma:sample_splitting}, we obtain that 
\begin{align*}
    T_1 = O_\mathbb{P}\left( \frac{\|\widehat{\xi}_d - \xi_d\|}{\sqrt{n}}\right) = o_\mathbb{P}(1/\sqrt{n}),
\end{align*}
where the last equality follows from our assumption that $\|\widehat{\xi}_d - \xi_d\| = o_\mathbb{P}(1)$. 
On the other hand, 
when ${\xi}_d$ and $\widehat{\xi}_d$ are contained in a Donsker class 
and $\|\widehat{\xi}_d - \xi_d\| = o_\mathbb{P}(1)$, 
by~\citet[Lemma 19.24]{van2000asymptotic}, we have that 
$T_1 = o_\mathbb{P}(1/\sqrt{n})$. 
Since $\mathbb{P}\xi_d - \psi_d = 0$, we have that 
\begin{align*}
    T_3 &= \mathbb{P}\left(\widehat{\xi}_d - \xi_d \right)
    = \mathbb{P} \Bigg \{
    \frac{\pi_1 d}{\widehat{\pi}_1}\frac{F_{1,m} - F_{1,\widehat{m}} }{ \widehat{f}_{1,\widehat{m}}} + \frac{\pi_0 (1 - d)}{\widehat{\pi}_0}\frac{F_{0,m} - F_{0,\widehat{m}} }{ \widehat{f}_{0, \widehat{m}}}
    + \widehat{m}_{d} - m_{d}
    \Bigg\} \\
    &= \mathbb{P} \Bigg \{d 
    \left( \widehat{m}_1 - m_1 - 
        \frac{\pi_1 }{\widehat{\pi}_1 }\frac{ F_{1,\widehat{m}} - F_{1,m} }{ \widehat{f}_{1,\widehat{m}}}
        \right)
        + (1 - d)
        \left( \widehat{m}_0 - m_0 - 
        \frac{\pi_0}{\widehat{\pi}_0}\frac{ F_{0,\widehat{m}} - F_{0,m}  }{ \widehat{f}_{0, \widehat{m}}}
        \right)
    \Bigg\}\\
&= R_{d} \left(\widehat{\mathbb{P}}, \mathbb{P} \right),
\end{align*}
where the second equality holds since $\mathbb{P}\xi_d = \psi_d = \mathbb{P} \{m_d\} $ and 
$\mathbb{E}_{Y|X,A=a}[\mathbb{1}\{Y\leq \widehat{m}_a(X)\}|X,A=a] = F_{a,\widehat{m}}(X)$.
Finally, using Lemma~\ref{corollary:R_d_product}, we obtain that  
\begin{align*}
    R_d(\widehat{\mathbb{P}}, \mathbb{P})
    &= \mathbb{P}\Bigg\{ 
    \frac{d\left(\widehat{m}_1 -m_1 \right)}{\widehat{\pi}_1  \widehat{f}_{1,\widehat{m}}} 
    \left(
    \left(\widehat{\pi}_1 - \pi_1 \right)  \widehat{f}_{1,\widehat{m}} 
    +  
    \left(  \widehat{f}_{1,\widehat{m}} - f_{1,m}\right) \pi_1
    - (\widehat{m}_1 - m_1)\frac{f'_{1,c} \pi_1 }{2}\right) \\
    & + \frac{(1 - d)\left(\widehat{m}_0 -m_0 \right)}{\widehat{\pi}_0  \widehat{f}_{0,\widehat{m}}} 
    \left(
    (\widehat{\pi}_0 - \pi_0)  \widehat{f}_{0,\widehat{m}} 
    +(\widehat{f}_{0,\widehat{m}} - f_{0,m})\pi_0
    - (\widehat{m}_0 - m_0) \frac{f'_{0,c} \pi_0 }{2 }\right) \Bigg\}\\
    &= O_{\mathbb{P}}\left(
    \sum_{a=0}^1 \left\|\widehat{m}_a - m_a\right\| \; 
    \left( \left\|\widehat{\pi}_a \widehat{f}_{a,\widehat{m}}  - \pi_a {f}_{a, m} \right\| 
    + \|\widehat{m}_a - m_a\| 
   \right) \right). 
\end{align*}
\end{proof}

\begin{proof}[Proof of Corollary~\ref{corollary:fixed-policy-corollary}]
Following from the proof of Theorem~\ref{thm:fixed_policy_estimator}, 
we have that 
$\widehat{\psi}_{d, \text{dr}} - \psi_d = T_1 + T_2 + T_3$ where 
$T_1 = o_{\mathbb{P}}(1/\sqrt{n})$. 
From Theorem~\ref{thm:variance}, 
we have that $\text{Var}\{\xi_d\} = \text{Var}\{\phi_d\} = \sigma_d^2 < +\infty$ since $\mathbb{P} \in \mathcal{P}$. 
By Central Limit Theorem, we have that $T_2 = o_\mathbb{P}(1/\sqrt{n})$ and  
\begin{align*}
    \sqrt{n} T_2 = \sqrt{n} (\mathbb{P}_n - \mathbb{P}) \xi_d(\mathbb{P}) 
    = \sqrt{n} \left( \frac{1}{n} \sum_{i=1}^n \xi_d(Z_i) -  \mathbb{E}[\xi_d(Z)] \right) \leadsto \mathcal{N}(0, \text{Var}(\xi_d)).
\end{align*}
Using Slutsky's theorem and the assumptions that ensure $T_3 = o_{\mathbb{P}}(1/\sqrt{n})$, 
we have 
\begin{align*}
    \sqrt{n}(\widehat{\psi}_{d, \text{dr}} - \psi_d)
    \leadsto \mathcal{N}(0, \sigma_d^2). 
\end{align*}
\end{proof}

\begin{proof}[Proof of Theorem~\ref{thm:learned_policy_estimator}]
For any $d \in \mathcal{D}$,
by the fact that $\mathbb{P}\xi_d = \psi_d$, %
we have that 
\begin{align*}
    \widehat{\psi}_{\widehat{d}^*, \text{dr}} -  \psi_{d^*}
    &= \mathbb{P}_n \widehat{\xi}_{\widehat{d}^*}
    -  \mathbb{P} {\xi}_{\widehat{d}^*}
    +  \mathbb{P} {\xi}_{\widehat{d}^*}
    - \mathbb{P} \xi_{d^*} \\
    &=\mathbb{P}_n \widehat{\xi}_{\widehat{d}^*}
    -  \mathbb{P} {\xi}_{\widehat{d}^*} + \mathbb{P} \gamma (\widehat{d}^* - d^*),
\end{align*}
where the second equality uses that 
\begin{align*}
    \mathbb{P} {\xi}_{\widehat{d}^*}
    - \mathbb{P} \xi_{d^*} 
    &= \mathbb{P} \Bigg \{
    \frac{A \left( \widehat{d}^*(X) - d^*(X)\right)}{{\pi_1}(X)}\frac{1/2 - \mathbb{1}\{Y \leq {m}_1(X) \}}{ {f}_{1,{m}}\left(X \right)} \\
&\qquad + \frac{(1-A)(d^*(X)-\widehat{d}^*(X))}{\pi_0(X)}\frac{1/2 - \mathbb{1}\{Y \leq {m}_0(X) \}}{ {f}_{0, {m}}\left(X\right)}
+ {m}_{\widehat{d}^*}(X) - m_{d^*}(X)
\Bigg\} \\
&= \mathbb{P}\left\{ m_1 (\widehat{d}^* - d^*)
+ m_0 (d^* - \widehat{d}^*) \right\}= \mathbb{P} \gamma (\widehat{d}^* - d^*).
\end{align*}
The second equality above comes from the fact that  $\mathbb{E}_{Y|X,A=a}[\mathbb{1}\{Y \leq m_a(X)\}] = 1/2$. 
Further, by Lemma~\ref{lemma:policy_diff} and the fact that $d^* = \mathbb{1}\{\gamma > 0\}$, we obtain that 
\begin{align*}
\mathbb{P} \gamma (\widehat{d}^* - d^*)
&\leq \mathbb{P} \gamma \mathbb{1}\{|\gamma| \leq |\widehat{\gamma} - \gamma|\}
\leq  \mathbb{P} |\widehat{\gamma} - \gamma|  \mathbb{1}\{|\gamma| \leq |\widehat{\gamma} - \gamma|\} \\
&\leq \|\widehat{\gamma} - \gamma\|_\infty \mathbb{P}\left(|\gamma| \leq \|\widehat{\gamma} - \gamma\|_\infty \right)
\lesssim  \|\widehat{\gamma} - \gamma\|^{1+\alpha}_\infty,
\end{align*}
where the last inequality follows from the margin condition.
On the other hand, 
similar to the proof of Theorem~\ref{thm:fixed_policy_estimator}, 
we use the decomposition 
$$
\mathbb{P}_n \widehat{\xi}_{\widehat{d}^*}
-  \mathbb{P} {\xi}_{\widehat{d}^*}
= \left(\mathbb{P}_n 
- \mathbb{P}\right) 
\left( \widehat{\xi}_{\widehat{d}^*} -  {\xi}_{{d}^*}\right)
+ \left( \mathbb{P}_n 
- \mathbb{P}\right)  {\xi}_{{d}^*} 
+ \mathbb{P} \left( \widehat{\xi}_{\widehat{d}^*}
-  {\xi}_{\widehat{d}^*} \right).
$$ 
By the fact that $\mathbb{P}\phi_{\widehat{d}^*} = \psi_{\widehat{d}^*} + \mathbb{P}\xi_{\widehat{d}^*} = 0$, we have that 
\begin{align*}
    \mathbb{P} \left( \widehat{\xi}_{\widehat{d}^*}
    -   {\xi}_{\widehat{d}^*} \right)
&= \mathbb{P} \Bigg \{
    \frac{\pi_1 \widehat{d}^*}{\widehat{\pi}_1}\frac{F_{1,m} - F_{1,\widehat{m}} }{ \widehat{f}_{1,\widehat{m}}} + \frac{\pi_0 (1 - \widehat{d}^*)}{\widehat{\pi}_0}\frac{F_{0,m} - F_{0,\widehat{m}} }{ \widehat{f}_{0, \widehat{m}}}
+ \widehat{m}_{\widehat{d}^*} - m_{\widehat{d}^*}
\Bigg\} \\
&= \mathbb{P} \Bigg \{\widehat{d}^* 
\left( \widehat{m}_1 - m_1 - 
    \frac{\pi_1 }{\widehat{\pi}_1 }\frac{ F_{1,\widehat{m}} - F_{1,m} }{ \widehat{f}_{1,\widehat{m}}}
    \right)
    + (1 - \widehat{d}^*)
    \left( \widehat{m}_0 - m_0 - 
    \frac{\pi_0}{\widehat{\pi}_0}\frac{ F_{0,\widehat{m}} - F_{0,m}  }{ \widehat{f}_{0, \widehat{m}}}
    \right)
\Bigg\}\\
&= R_{\widehat{d}^*} \left(\widehat{\mathbb{P}}, \mathbb{P} \right) = 
O_{\mathbb{P}}\left(
    \sum_{a=0}^1 \left\|\widehat{m}_a - m_a\right\| \; 
    \left( \left\|\widehat{\pi}_a \widehat{f}_{a,\widehat{m}} - \pi_a {f}_{a, m}\right\| 
    + \|\widehat{m}_a - m_a\| 
    \right) \right).\\
\end{align*}
Finally, for the %
term 
$\left(\mathbb{P}_n 
- \mathbb{P}\right) 
\left( \widehat{\xi}_{\widehat{d}^*} -  {\xi}_{{d}^*}\right)$: 
when $\widehat{\xi}_{\widehat{d}^*}$ 
is estimated from a separate sample $Z^{n,0}$ from the empirical measure over $Z^n$, by Lemma~\ref{lemma:sample_splitting}
we have that 
\begin{align*}
    \left(\mathbb{P}_n 
- \mathbb{P}\right) 
\left( \widehat{\xi}_{\widehat{d}^*} -  {\xi}_{{d}^*}\right)
= O_{\mathbb{P}}\left(\frac{\|\widehat{\xi}_{\widehat{d}^*} -  {\xi}_{{d}^*}\|}{\sqrt{n}} \right)
= o_\mathbb{P}(1/\sqrt{n}),
\end{align*}
where the last equality follows from our assumption that $\|\widehat{\xi}_{\widehat{d}^*} -  {\xi}_{{d}^*}\| = o_\mathbb{P}(1)$. 
When $\xi_{{d}^*}$ and $\widehat{\xi}_{\widehat{d}^*}$ are contained in a Donsker class
and $\|\widehat{\xi}_{\widehat{d}^*} - \xi_{{d}^*}\|=o_\mathbb{P}(1)$, by~\citet[Lemma 19.24]{van2000asymptotic}, we have that 
$\left(\mathbb{P}_n 
- \mathbb{P}\right) 
\left( \widehat{\xi}_{\widehat{d}^*} -  {\xi}_{{d}^*}\right) = o_\mathbb{P}(1/\sqrt{n})$. 
\end{proof}

\begin{proof}[Proof of Corollary~\ref{corollary:learned-policy-corollary}]
From Theorem~\ref{thm:learned_policy_estimator}, 
we have that 
\begin{align*}
    \widehat{\psi}_{\widehat{d}^*, \text{dr}} -  \psi_{d^*} = 
    \left(\mathbb{P}_n 
        - \mathbb{P}\right) {\xi}_{{d}^*}
    + O_{\mathbb{P}}\Big(
    & o_\mathbb{P}(1/\sqrt{n}) + \\
    & \underbrace{ \|\widehat{\gamma} - \gamma\|^{1+\alpha}_\infty
     + 
    \sum_{a=0}^1 \left\|\widehat{m}_a - m_a\right\| \; 
    \left( 
    \|\widehat{\pi}_a \widehat{f}_{a,\widehat{m}} - \pi_a {f}_{a, m} \| 
    + \|\widehat{m}_a - m_a\| 
    \right)}_{R} \Big).
\end{align*}
The conditions in Corollary~\ref{corollary:learned-policy-corollary}
ensure that 
$R = o_\mathbb{P}(1/\sqrt{n})$, 
which gives that 
\begin{align*}
    \sqrt{n}(\widehat{\psi}_{\widehat{d}^*, \text{dr}} - \psi_{d^*})
     \leadsto \mathcal{N} \left(0, \sigma_{d^*}^2\right),
\end{align*}
since $\text{Var}\{\xi_{d^*}\} = \text{Var}\{\phi_{d^*}\} = \sigma_{d^*}^2$.
\end{proof}

\subsection{Proofs in Section~\ref{sec:learning-gamma}}
\begin{proof}[Proof of Theorem~\ref{thm:gamma-error}]
    We begin with a decomposition: 
    \begin{align*}
        &\wh \gamma_\dr(x) - \gamma(x)= \sum_{i=1}^n w_i(x; X^n) \wh g(Z_i)
        - \gamma(x) \\
        =& \left( \sum_{i=1}^n w_i(x; X^n)  g(Z_i) - \gamma(x)\right)
        + \sum_{i=1}^n w_i(x; X^n) h(X_i)
        + \sum_{i=1}^n w_i(x; X^n) \left( \wh g(Z_i) - g(Z_i) - h(X_i)\right)\\
        =& \left(\wt \gamma(x) - \gamma(x)\right)
        + \sum_{i=1}^n w_i(x; X^n) h(X_i)
        + \sum_{i=1}^n w_i(x; X^n) \left(\wh g(Z_i) - g(Z_i) - h(X_i)\right),
    \end{align*}
    where $h(x) = \E[\wh g(Z) - g(Z) | D^n, X=x]$
    and $D^n = (D_1, D_2)$ is the training data for obtaining $\wh g$.
    
    To bound $h(x)$, %
    we obtain that 
    \begin{align*}
       &h(x) = 
       \mathbb{E}[\wh g(Z) - g(Z)|D^n, X=x]
        = \mathbb{E}[\wh g(Z)|D^n, X=x]
        - \gamma(x)\\
        =& \left(\wh m_1(x) - m_1(x) - \frac{\pi_1(x)}{\wh \pi_1(x)} \frac{F_{1,\wh m}(x) - F_{1, m}(x)}{\wh f_{1,\wh m}(x)} \right)
        - \left(\wh m_0(x) - m_0(x) - \frac{\pi_0(x)}{\wh \pi_0(x)} \frac{F_{0, \wh m}(x) - F_{0, m}(x)}{\wh f_{0,\wh m}(x)} \right)\\
        =& 
        \left(\left(\wh{m}_1(x) -m_1(x) \right)
        \left(1 - \frac{\pi_1(x)}{\wh{\pi}_1(x)}  \frac{f_{1,m}(x)}
        {\wh{f}_{1, \wh{m}}(x)}\right) 
        - \frac{f'_{1,c}(x) \pi_1(x) (\wh{m}_1(x) - m_1(x))^2}{2\wh{\pi}_1(x)\wh{f}_{1,\wh{m}}(x)}\right)\\
        &\quad -
        \left(\left(\wh{m}_0(x) -m_0(x) \right)
        \left(1 - \frac{\pi_0(x)}{\wh{\pi}_0(x)}  \frac{f_{0,m}(x)}
        {\wh{f}_{0, \wh{m}}(x)}\right) 
        - \frac{f'_{0,c}(x) \pi_0(x) (\wh{m}_0(x) - m_0(x))^2}{2\wh{\pi}_0(x)\wh{f}_{0,\wh{m}}(x)}\right)\\
        \lesssim& 
        \sum_{a=0}^1 |\wh m_a(x) - m_a(x)| 
        \left(|\wh \pi_a(x) \wh f_{a,\wh m}(x) -  \pi_a(x)  f_{a, m}(x)| + |\wh m_a(x) - m_a(x)|\right),
    \end{align*}
    where $f'_{a,c}(x) \leq L$ is the derivative of $f_a(y \mid x)$
    at $y = c_a(x)$ 
    for some value $c_a(x)$ between $m_a(x)$ and $\wh{m}_a(x)$
    and the third equality follows similarly to the proof of Corollary~\ref{corollary:R_d_product}.
    This gives that 
    \begin{align*}
        &\sum_{i=1}^n w_i(x; X^n) h(X_i) \\
        \lesssim& \sum_{i=1}^n |w_i(x; X^n)|
        \left(\sum_{a=0}^1 |\wh m_a(X_i) - m_a(X_i)| 
        \left(|\wh \pi_a(X_i) \wh f_{a,\wh m}(X_i) -  \pi_a(X_i)  f_{a, m}(X_i)| + |\wh m_a(X_i) - m_a(X_i)|\right)\right)
        \\
         =& \sum_{a=0}^1 \sum_{i=1}^n |w_i(x; X^n)| |\wh m_a(X_i) - m_a(X_i)| \left(\left|\pi_a(X_i) f_{a,m}(X_i) - \wh \pi_a(X_i) \wh f_{a, \wh m}(X_i) \right|+ |\wh m_a(X_i) - m_a(X_i)|\right)\\
        =& \sum_{a=0}^1 \sum_{i=1}^n |w_i(x; X^n)| |\wh m_a(X_i) - m_a(X_i)| \cdot \left|\pi_a(X_i) f_{a,m}(X_i) - \wh \pi_a(X_i) \wh f_{a, \wh m}(X_i) \right|\\
        &\quad +\sum_{a=0}^1 \sum_{i=1}^n |w_i(x; X^n)| \cdot |\wh m_a(X_i) - m_a(X_i)| \cdot |\wh m_a(X_i) - m_a(X_i)|
        \\
        \leq& \sum_{a=0}^1 \sqrt{\sum_{i=1}^n |w_i(x; X^n)| |\wh m_a(X_i) - m_a(X_i)|^2 }   \cdot \sqrt{\sum_{i=1}^n |w_i(x; X^n)| \left|\pi_a(X_i) f_{a,m}(X_i) - \wh \pi_a(X_i) \wh f_{a, \wh m}(X_i) \right|^2}\\
        &+\quad \sum_{a=0}^1 \sqrt{\sum_{i=1}^n |w_i(x; X^n)| |\wh m_a(X_i) - m_a(X_i)|^2 }   \cdot \sqrt{\sum_{i=1}^n |w_i(x; X^n)| \left|\wh m_a(X_i) - m_a(X_i)\right|^2}\\
        =& \left(\sum_{i=1}^n |w_i(x; X^n)|\right) \cdot \left( \sum_{a=0}^1 \|\wh m_a - m_a\|_w \left(\|\wh \pi_a \wh f_{a, \wh m} - \pi_a f_{a,m}\|_w + \|\wh m_a - m_a\|_w\right) \right).
    \end{align*}

    To bound $G(x) := \sum_{i=1}^n w_i(x; X^n) \left(\wh g(Z_i) - g(Z_i) - h(X_i)\right)$, 
    we observe that by definition, 
    \begin{align*}
         \E [\wh g(Z_i) - g(Z_i) - h(X_i) | D^n, X^n]
        = 0,
    \end{align*} 
    where $X^n$ are the covariates in $D_3$. 
    We also have 
    \begin{align*}
        \text{Var}\{G(x) | D^n, X^n\}
        &= \text{Var}\left\{\sum_{i=1}^n w_i(x; X^n) \left(\wh g(Z_i) - g(Z_i) - h(X_i) \right)| D^n, X^n \right\}\\
        &= \sum_{i=1}^n w_i(x; X^n)^2 \text{Var}\left\{\wh g(Z_i) - g(Z_i) - h(X_i) | D^n, X^n \right\}\\
        &= \sum_{i=1}^n w_i(x; X^n)^2 \text{Var}\left\{\wh g(Z_i) - g(Z_i) | D^n, X^n \right\}\\
        &\leq  \|\wh g - g\|_{w^2}^2 \sum_{i=1}^n w_i(x; X^n)^2,
    \end{align*}
    where the third equality holds since 
    $\text{Var}\left\{\wh g(Z_i) - g(Z_i) | D^n, X^n \right\} = \E[(\wh g(Z_i) - g(Z_i))^2 | D^n, X^n ] - \E[\wh g(Z_i) - g(Z_i) | D^n, X^n ]^2 \leq \E[(\wh g(Z_i) - g(Z_i))^2 | D^n, X^n ]$.
    We note that 
    \begin{align*}
        \E[(\wt \gamma(x) - \gamma(x))^2] 
        &= \E\left[\left(\sum_{i=1}^n w_i(x; X^n) \left(g(Z_i) - \gamma(X_i)\right)
        + \sum_{i=1}^n w_i(x; X^n) \gamma(X_i)
        - \gamma(x)\right)^2 \right] 
        \\
        &=\E\left[\left(\sum_{i=1}^n w_i(x; X^n) \left(g(Z_i)  - \gamma(X_i)\right) \right)^2\right]
        + \E\left[\left(\sum_{i=1}^n w_i(x; X^n) \gamma(X_i) - \gamma(x)\right)^2\right]\\
        &= \E \left[\sum_{i=1}^n w_i(x; X^n)^2 \text{Var}\{g(Z_i)| X_i\} \right] 
        + \E\left[\left(\sum_{i=1}^n w_i(x; X^n) \gamma(X_i) - \gamma(x)\right)^2\right]
        \\
        &\geq \sigma^2_{\min} \sum_{i=1}^n \E \left[ w_i(x; X^n)^2 \right],
    \end{align*}
    where the second and third equality hold since $\E[g(Z_i)|X^n] = \gamma(X_i)$ 
    and all samples are independent.  
    Putting it altogether, using Markov's inequality, we obtain
    \begin{align*}
        &\mathbb{P}\left( \frac{|G(x)|}{\|\wh g - g\|_{w^2} \sqrt{\E[(\wt \gamma(x) - \gamma(x))^2] }} \geq t\right)
        = \E\left[ \mathbb{P}\left( \frac{|G(x)|}{\|\wh g - g\|_{w^2}\sqrt{\E[(\wt \gamma(x) - \gamma(x))^2] }} \geq t \Big| D^n, X^n\right) \right]\\
        \leq& \frac{1}{t^2} \cdot \E\left[\frac{\text{Var}\{G(x) | D^n, X^n\}}{ \|\wh g - g\|^2_{w^2}\E[(\wt \gamma(x) - \gamma(x))^2] } \right] \leq \frac{1}{\sigma_{\min}^2 t^2}.
    \end{align*} 
    Thus, we have that 
    $G(x) = O_\mathbb{P}(\|\wh g - g\|_{w^2} \sqrt{\E[(\wt \gamma(x) - \gamma(x))^2]})$
    and 
    \begin{align*}
        G(x)^2 = O_\mathbb{P}\left(\|\wh g - g\|^2_{w^2} {\E[(\wt \gamma(x) - \gamma(x))^2]}\right). %
    \end{align*}
    The proof completes by realizing that 
    $(a+b+c)^2 \leq 3 a^2 + 3 b^2 + 3 c^2$.
    \end{proof}

\section{Auxiliary Lemmas} \label{appendix:auxillary}
Below are the auxiliary lemmas we have used to prove the results in this paper. 
For completeness, we include their proofs here as well. 

\begin{lemma}\citep[Lemma 1]{kennedy2018sharp}\label{lemma:policy_diff}
Let $\widehat{f}, f$ be functions taking any real values. Then
\begin{align*}
    |\mathbb{1}\{\widehat{f} > 0\} - \mathbb{1}\{f > 0\}|
    \leq \mathbb{1}\{|f| \leq |\widehat{f} - f|\}. 
\end{align*}
\end{lemma}

\begin{proof}
It follows that 
\begin{align*}
    |\mathbb{1}\{\widehat{f} > 0\} - \mathbb{1}\{f > 0\}|
    = \mathbb{1}\{f, \widehat{f} \text{ have opposite signs}\}
    \leq \mathbb{1}\{|f| \leq |\widehat{f} - f|\},
\end{align*}
since $|\widehat{f} - f| = |\widehat{f}| + |f|$
when $f, \widehat{f}$ have opposite signs.
\end{proof}

\begin{lemma}\citep[Lemma 2]{kennedy2018sharp}\label{lemma:sample_splitting}
Let $\mathbb{P}_n$ denote the empirical measure over $Z^n = (Z_1, \ldots, Z_n)$ and $\widehat{\phi}$ be a function estimated from a sample $Z^{n,0} = (Z^0_{1}, \ldots, Z^0_{n})$, which is independent of $Z^n$, then for any function $\phi$, 
\begin{align*}
    \left(\mathbb{P}_n - \mathbb{P} \right)(\widehat{\phi} - \phi) = O_{\mathbb{P}}\left(\frac{\|\widehat{\phi} - \phi\|}{\sqrt{n}}\right).
\end{align*}
\end{lemma}

\begin{proof}

First, we notice that 
\begin{align*}
    \text{Var}\left\{ {\left(\mathbb{P}_n - \mathbb{P}\right)(\widehat{\phi} - \phi)}\big| Z^{n,0} \right\} = \text{Var}\left\{ {\mathbb{P}_n (\widehat{\phi} - \phi)}\big| Z^{n,0} \right\} = \frac{1}{n}\text{Var}\left\{\widehat{\phi} - \phi|Z^{n,0}\right\} \leq \frac{\|\widehat{\phi} - \phi\|^2_{2}}{n},
\end{align*}
where the first equality is true since conditioned on $Z^{n,0}$, $\mathbb{P} (\widehat{\phi} - \phi)$ is a constant.
Then by applying the law of total expectation and Chebyshev's inequality, we obtain that 
\begin{align*}
    \mathbb{P}\left(\frac{ | \left(\mathbb{P}_n - \mathbb{P}\right)(\widehat{\phi} - \phi) |}{\|\widehat{\phi} - \phi\|_2/\sqrt{n}} \geq t \right) &= \mathbb{E}\left\{\mathbb{P}\left(\frac{|\left(\mathbb{P}_n - \mathbb{P}\right)(\widehat{\phi} - \phi)|}{\|\widehat{\phi} - \phi\|_2/\sqrt{n}} \geq t \bigg| Z^{n,0} \right)\right\} \\
    &\leq \mathbb{E}\left\{\frac{\text{Var}\left\{ {\left(\mathbb{P}_n - \mathbb{P}\right)(\widehat{\phi} - \phi)}\big| Z^{n,0} \right\}}{\|\widehat{\phi} - \phi\|_2^2t^2/n}\right\} \leq \frac{1}{t^2},
\end{align*}
where we have utilized the fact that $\mathbb{E}\left\{ {\left(\mathbb{P}_n - \mathbb{P}\right)(\widehat{\phi} - \phi)}\bigg| Z^{n,0} \right\} = \mathbb{E}\{\widehat{\phi} - \phi|Z^{n,0}\} - \mathbb{P}(\widehat{\phi} - \phi)= 0$.
\end{proof}

\end{document}